\documentclass[11pt,letterpaper]{article}
\usepackage[lmargin=1.0in,rmargin=1.0in,bottom=1.0in,top=1.0in,twoside=False]{geometry}

\usepackage{fullpage,amssymb,amsmath}
\usepackage{graphicx}
\usepackage{enumerate}
\usepackage{tikz}
\usetikzlibrary{shapes}

\usepackage[T1]{fontenc}

 \usepackage{xcolor}
 \usepackage{mathtools}
\usepackage{microtype}
\usepackage{amsfonts}
\usepackage{comment}
\usepackage[english]{babel}
\usepackage{mathrsfs}
\usepackage[ruled,vlined]{algorithm2e}

\usepackage{trimspaces}
\usepackage{nccfoots}
\usepackage{setspace}
\usepackage{inconsolata}
\usepackage{libertine}
\usepackage[absolute]{textpos}

\usepackage{enumitem}
\usepackage{todonotes}

\usepackage{longtable}

\definecolor{blue}{rgb}{0.1,0.2,0.5}
\definecolor{brown}{rgb}{0.6,0.6,0.2}
\usepackage[ocgcolorlinks, linkcolor={blue}, citecolor={brown}]{hyperref}

\usepackage{comment}

\usepackage[amsmath,thmmarks,hyperref]{ntheorem}
\usepackage{cleveref}

\crefformat{page}{#2page~#1#3}%
\Crefformat{page}{#2Page~#1#3}%
\crefformat{equation}{#2(#1)#3}%
\Crefformat{equation}{#2(#1)#3}%
\crefformat{figure}{#2Figure~#1#3}%
\Crefformat{figure}{#2Figure~#1#3}%
\crefformat{section}{#2Section~#1#3}
\Crefformat{section}{#2Section~#1#3}
\crefformat{chapter}{#2Chapter~#1#3}
\Crefformat{chapter}{#2Chapter~#1#3}
\crefformat{chapter*}{#2Chapter~#1#3}
\Crefformat{chapter*}{#2Chapter~#1#3}
\crefformat{part}{#2Part~#1#3}
\Crefformat{part}{#2Part~#1#3}
\crefformat{enumi}{#2(#1)#3}
\Crefformat{enumi}{#2(#1)#3}

\usepackage{enumerate}

\usepackage{latexsym}

\theoremnumbering{arabic}
\theoremstyle{plain}
\theoremsymbol{}
\theorembodyfont{\itshape}
\theoremheaderfont{\normalfont\bfseries}
\theoremseparator{.}

\newtheorem{theorem}{Theorem}
\crefformat{theorem}{#2Theorem~#1#3}
\Crefformat{theorem}{#2Theorem~#1#3}

\newcommand{\newtheoremwithcrefformat}[2]{%
  \newtheorem{#1}[theorem]{#2}%
  \crefformat{#1}{##2\MakeUppercase#1~##1##3}%
  \Crefformat{#1}{##2\MakeUppercase#1~##1##3}%
}
\newcommand{\newseptheoremwithcrefformat}[2]{%
  \newtheorem{#1}{#2}%
  \crefformat{#1}{##2\MakeUppercase#1~##1##3}%
  \Crefformat{#1}{##2\MakeUppercase#1~##1##3}%
}

\newtheoremwithcrefformat{lemma}{Lemma}
\newtheoremwithcrefformat{proposition}{Proposition}
\newtheoremwithcrefformat{observation}{Observation}
\newtheoremwithcrefformat{conjecture}{Conjecture}
\newtheoremwithcrefformat{corollary}{Corollary}
\newseptheoremwithcrefformat{claim}{Claim}
\theorembodyfont{\upshape}
\newtheoremwithcrefformat{example}{Example}
\newtheoremwithcrefformat{remark}{Remark}
\newseptheoremwithcrefformat{definition}{Definition}
\newseptheoremwithcrefformat{question}{Question}

\theoremstyle{nonumberplain}
\theoremheaderfont{\scshape}
\theorembodyfont{\normalfont}
\theoremsymbol{\ensuremath{\square}}
\newtheorem{proof}{Proof}

\theoremsymbol{\ensuremath{\lrcorner}}
\newtheorem{clproof}{Proof of Claim}

\def\cqedsymbol{\ifmmode$\lrcorner$\else{\unskip\nobreak\hfil
\penalty50\hskip1em\null\nobreak\hfil$\lrcorner$
\parfillskip=0pt\finalhyphendemerits=0\endgraf}\fi}

\tikzset{
    position/.style args={#1:#2 from #3}{
        at=(#3.#1), anchor=#1+180, shift=(#1:#2)
    }
}

\newcommand{\N}{\mathbb{N}}

\newcommand{\R}{\mathbb{R}}

\newcommand{\Oh}{\mathcal{O}}

\newcommand{\OPT}{\mathrm{OPT}}

\let\originalleft\left
\let\originalright\right
\renewcommand{\left}{\mathopen{}\mathclose\bgroup\originalleft}
\renewcommand{\right}{\aftergroup\egroup\originalright}

\renewcommand{\leq}{\leqslant}
\renewcommand{\geq}{\geqslant}

\renewcommand{\setminus}{-}

\newcommand{\OCT}{\textsc{Odd Cycle Transversal}\xspace}

\newcommand{\MWLHS}{\textsc{Max Partial $H$-Coloring}\xspace}
\newcommand{\MMWLHS}{\textsc{Max Partial $H$-Multicoloring}\xspace}
\newcommand{\LHomo}{\textsc{List $H$-Coloring}\xspace}

\newcommand{\wh}[1]{\widehat{#1}}
\newcommand{\power}{\mathsf{Pow}^\star}
\newcommand{\Mod}{\mathsf{Mod}}
\newcommand{\Quo}{\mathsf{Quo}}

\newcommand{\partto}{\rightharpoonup}
\newcommand{\wei}{\mathsf{rev}}
\newcommand{\Ff}{\mathcal{F}}
\newcommand{\Rr}{\mathcal{R}}
\newcommand{\Tt}{\mathcal{T}}
\newcommand{\dom}{\mathsf{dom}\,}

\newcommand{\Ramsey}{\mathsf{Ramsey}}

\newcommand{\defproblem}[3]{
  \vspace{3mm}
  \noindent\fbox{
  \begin{minipage}{0.97\textwidth}  
  #1 \\ 
  {\bf{Input:}} #2  \\
  {\bf{Output:}} #3
  \end{minipage}
  }
  \vspace{3mm}
}

\usepackage{todonotes}

\newcommand{\antennaoneandhalf}{\bullet{}\hspace{-0.3em}-\hspace{-0.1em}\circ{}}
\newcommand{\universal}{\bullet}

\newcommand{\ante}[1]{{#1}^{\antennaoneandhalf}}
\newcommand{\univ}[1]{{#1}^{\universal}}

\begin{document}

\title{Finding large $H$-colorable subgraphs in hereditary graph classes}

\author{
Maria Chudnovsky\thanks{
Princeton University, Princeton, NJ 08544, \texttt{mchudnov@math.princeton.edu}. 
This material is based upon work supported in part by the U. S. Army
Research Office under grant number W911NF-16-1-0404, and by  NSF grant DMS-1763817.
}
\and
Jason King\thanks{
Princeton University, Princeton, NJ 08544, \texttt{jtking@princeton.edu}
}
\and
Micha\l{}~Pilipczuk\thanks{
  Institute of Informatics, University of Warsaw, Poland, \texttt{michal.pilipczuk@mimuw.edu.pl}.
  This work is 
a part of project TOTAL that has received funding from the European Research Council (ERC) 
under the European Union's Horizon 2020 research and innovation programme (grant agreement No.~677651).
}
\and
Pawe\l{} Rz\k{a}\.{z}ewski\thanks{
  Faculty of Mathematics and Information Science, Warsaw University of Technology, Poland, and Institute of
Informatics, University of Warsaw, Poland, \texttt{p.rzazewski@mini.pw.edu.pl}.
Supported by Polish National Science Centre grant no. 2018/31/D/ST6/00062.
}
\and
Sophie Spirkl\thanks{
Princeton University, Princeton, NJ 08544, \texttt{sspirkl@math.princeton.edu}.
This material is based upon work supported by the National Science Foundation under Award No. DMS1802201.
}
}

\begin{titlepage}
\def\thepage{}
\thispagestyle{empty}
\maketitle

\begin{textblock}{20}(0, 11.7)
\includegraphics[width=40px]{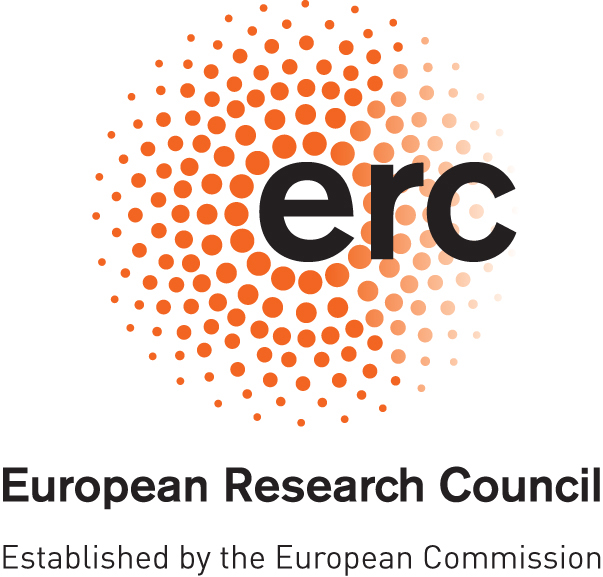}%
\end{textblock}
\begin{textblock}{20}(-0.25, 12.1)
\includegraphics[width=60px]{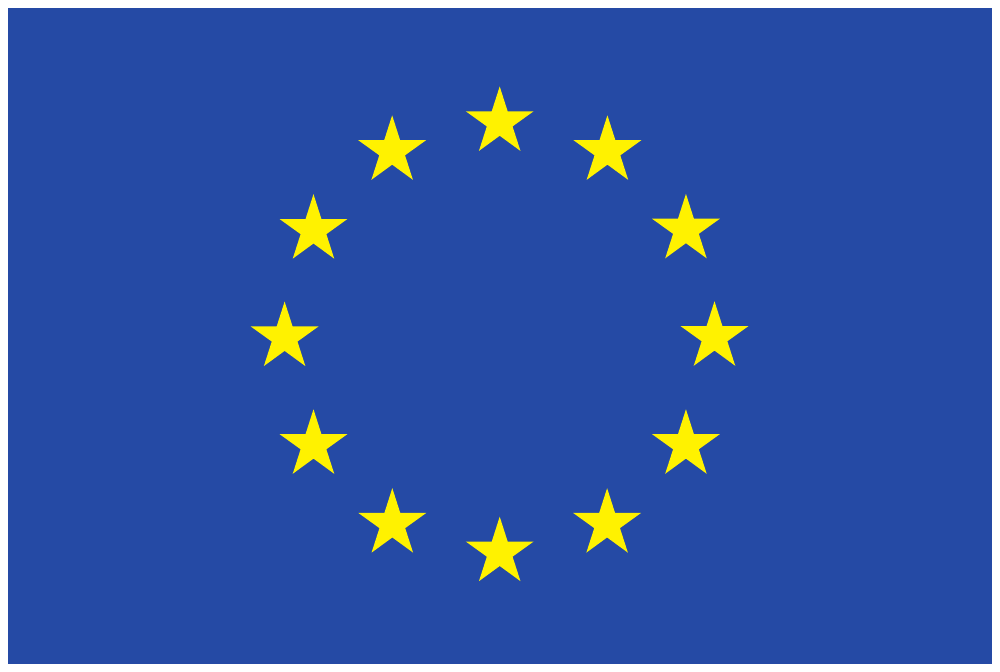}%
\end{textblock}

We study the \MWLHS problem: given a graph $G$, find the largest induced subgraph of $G$ that admits a homomorphism into $H$, where $H$ is a fixed pattern graph without loops.
Note that when $H$ is a complete graph on $k$ vertices, the problem reduces to finding the largest induced $k$-colorable subgraph, which for $k=2$ is equivalent (by complementation) to \OCT.

We prove that for every fixed pattern graph $H$ without loops, \MWLHS can be solved:
 \begin{itemize}[nosep]
  \item in $\{P_5,F\}$-free graphs in polynomial time, whenever $F$ is a threshold graph;
  \item in $\{P_5,\textrm{bull}\}$-free graphs in polynomial time;
  \item in $P_5$-free graphs in time $n^{\Oh(\omega(G))}$;
  \item in $\{P_6,\textrm{1-subdivided claw}\}$-free graphs in time $n^{\Oh(\omega(G)^3)}$.
 \end{itemize}
Here, $n$ is the number of vertices of the input graph $G$ and $\omega(G)$ is the maximum size of a clique in~$G$.
Furthermore, by combining the mentioned algorithms for $P_5$-free and for $\{P_6,\textrm{1-subdivided claw}\}$-free graphs with a simple branching procedure, we obtain subexponential-time algorithms for \MWLHS in these classes of graphs.

Finally, we show that even a restricted variant of \MWLHS is $\mathsf{NP}$-hard in the considered subclasses of $P_5$-free graphs, if we allow loops on $H$.

\end{titlepage}

\section{Introduction}\label{sec:intro}

Many computational graph problems that are ($\mathsf{NP}$-)hard in general become tractable in restricted classes of input graphs.
In this work we are interested in {\em{hereditary}} graph classes, or equivalently classes defined by forbidding induced subgraphs. 
For a set of graphs $\Ff$, we say that a graph $G$ is {\em{$\Ff$-free}} if $G$ does not contain any induced subgraph isomorphic to a graph from $\Ff$.
By forbidding different sets $\Ff$ we obtain graph classes with various structural properties, which can be used in the algorithmic context.
This highlights an interesting interplay between structural graph theory and algorithm design.

Perhaps the best known example of this paradigm is the case of the {\sc{Maximum Independent Set}} problem: given a graph $G$, find the largest set of pairwise non-adjacent vertices in $G$.
It is known that the problem is $\mathsf{NP}$-hard on $F$-free graphs unless $F$ is a forest whose every component is a path or a subdivided claw~\cite{Alekseev82}; here, a {\em{claw}} is a star with $3$ leaves.
However, the remaining cases, when $F$ is a {\em{subdivided claw forest}}, remain largely unexplored despite significant effort.
Polynomial-time algorithms have been given for $P_5$-free graphs~\cite{LokshtanovVV14}, $P_6$-free graphs~\cite{GrzesikKPP19}, claw-free graphs~\cite{Minty80,Sbihi80}, and fork-free graphs~\cite{Alekseev04,LozinM08}.
While the complexity status in all the other cases remains open, it has been observed that relaxing the goal of polynomial-time solvability leads to positive results in a larger generality.
For instance, for every $t\in \N$, {\sc{Maximum Independent Set}} can be solved in time $2^{\Oh(\sqrt{tn\log n})}$ in $P_t$-free graphs~\cite{BacsoLMPTL19}. 
The existence of such a {\em{subexponential-time algorithm}} for $F$-free graphs is excluded under the Exponential Time Hypothesis whenever $F$ is not a subdivided claw forest (see e.g. the discussion in~\cite{NovotnaOPRLW19}), 
which shows a qualitative difference between the negative and the potentially positive cases.
Also, Chudnovsky et al.~\cite{ChudnovskyPPT20} recently gave a quasi-polynomial-time approximation scheme (QPTAS) for {\sc{Maximum Independent Set}} in $F$-free graphs, for every fixed subdivided claw forests $F$.

The abovementioned positive results use a variety of structural techniques related to the considered hereditary graph classes, for instance: 
the concept of {\em{Gy\'arf\'as path}} that gives useful separators in $P_t$-free graphs~\cite{BacsoLMPTL19,Brause17,ChudnovskyPPT20}, 
the dynamic programming approach based on potential maximal cliques~\cite{LokshtanovVV14,GrzesikKPP19}, 
or structural properties of claw-free and fork-free graphs that relate them to line graphs~\cite{LozinM08,Minty80,Sbihi80}.
Some of these techniques can be used to give algorithms for related problems, which can be expressed as looking for the largest (in terms of the number of vertices) induced subgraph satisfying a fixed property.
For {\sc{Maximum Independent Set}} this property is being edgeless, but for instance the property of being acyclic corresponds to the {\sc{Maximum Induced Forest}} problem, which by complementation is equivalent to
{\sc{Feedback Vertex Set}}. Work in this direction so far focused on properties that imply bounded treewidth~\cite{ACPRS20,FominTV15} or, more generally, that imply sparsity~\cite{NovotnaOPRLW19}.

A different class of problems that admits an interesting complexity landscape on hereditary graphs classes are coloring problems.
For fixed $k\in \N$, the {\sc{$k$-Coloring}} problem asks whether the input graph admits a proper coloring with $k$ colors.
For every $k\geq 3$, the problem is $\mathsf{NP}$-hard on $F$-free graphs unless $F$ is a forest of paths (a {\em{linear forest}})~\cite{GolovachPS14}.
The classification of the remaining cases is more advanced than in the case of {\sc{Maximum Independent Set}}, but not yet complete.
On one hand, Ho\`ang et al.~\cite{kcolp5free} showed that for every fixed $k$, {\sc{$k$-Coloring}} is polynomial-time solvable on $P_5$-free graphs.
On the other hand, the problem becomes $\mathsf{NP}$-hard already on $P_6$-free graphs for all $k\geq 5$~\cite{Huang16}.
The cases $k=3$ and $k=4$ turn out to be very interesting. {\sc{$4$-Coloring}} is polynomial-time solvable on $P_6$-free graphs~\cite{SpirklCZ19} and $\mathsf{NP}$-hard in $P_7$-free graphs~\cite{Huang16}.
While there is a polynomial-time algorithm for {\sc{$3$-Coloring}} in $P_7$-free graphs~\cite{BonomoCMSSZ18}, the complexity status in $P_t$-free graphs for $t\geq 8$ remains open.
However, relaxing the goal again leads to positive results in a wider generality: 
for every $t\in \N$, there is a subexponential-time algorithm with running time $2^{\Oh(\sqrt{tn\log n})}$ for {\sc{$3$-Coloring}} in $P_t$-free graphs~\cite{GroenlandORSSS19},
and there is also a polynomial-time algorithm that given a $3$-colorable $P_t$-free graph outputs its proper coloring with $\Oh(t)$ colors~\cite{ChudnovskySSSZ19}.

We are interested in using the toolbox developed for coloring problems in $P_t$-free graphs to the setting of finding maximum induced subgraphs with certain properties.
Specifically, consider the following {\sc{Maximum Induced $k$-Colorable Subgraph}} problem: given a graph $G$, find the largest induced subgraph of $G$ that admits a proper coloring with $k$ colors.
While this problem clearly generalizes {\sc{$k$-Coloring}}, for $k=1$ it boils down to {\sc{Maximum Independent Set}}.
For $k=2$ it can be expressed as {\sc{Maximum Induced Bipartite Subgraph}}, which by complementation is equivalent to the well-studied \OCT problem: find the smallest subset of vertices that intersects
all odd cycles in a given graph. While polynomial-time solvability of \OCT on $P_4$-free graphs (also known as {\em{cographs}}) follows from the fact that these graphs have bounded cliquewidth (see~\cite{CourcelleMR00}),
it is known that the problem is $\mathsf{NP}$-hard in $P_6$-free graphs~\cite{DabrowskiFJPR19}. The complexity status of \OCT in $P_5$-free graphs remains open~\cite[Problem 4.4]{DBLP:journals/dagstuhl-reports/ChudnovskyPS19}: resolving this question was the original motivation of our work.

\paragraph*{Our contribution.} Following the work of Groenland et al.~\cite{GroenlandORSSS19}, we work with a very general form of coloring problems, defined through homomorphisms.
For graphs $G$ and $H$, a {\em{homomorphism}} from $G$ to~$H$, or an {\em{$H$-coloring}} of $G$, is a function $\phi\colon V(G)\to V(H)$ such that for every edge $uv$ in $G$, we have $\phi(u)\phi(v)\in E(H)$.
We study the \MWLHS problem defined as follows: given a graph~$G$, find the largest induced subgraph of $G$ that admits an $H$-coloring.
Note that if $H$ is the complete graph on $k$ vertices, then an $H$-coloring is simply a proper coloring with $k$ colors, hence this formulation generalizes the {\sc{Maximum Induced $k$-Colorable Subgraph}} problem.
We will always assume that the pattern graph~$H$ does not have loops, hence an $H$-coloring is always a proper coloring with $|V(H)|$ colors.

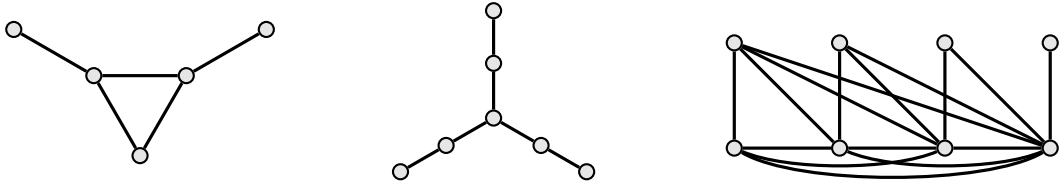
\begin{figure}[htb]
 \begin{center}
 \begin{tikzpicture}
  
   \tikzstyle{vertex}=[circle,thick,draw=black,fill=gray!20,minimum size=0.2cm,inner sep=0pt]

   \begin{scope}[shift={(-5,0)}]
   \node[vertex] (x) at (0,-0.5) {};
   
   \node[vertex, position=60:1 from x] (u1) {};
   \node[vertex, position=120:1 from x] (u2) {};
   
   \node[vertex, position=30:1 from u1] (v1) {};
   \node[vertex, position=150:1 from u2] (v2) {};
   
     \draw[very thick] (v1) -- (u1) -- (u2) -- (v2);
     \draw[very thick] (u1) -- (x) -- (u2);
   
   \end{scope}
   
   \begin{scope}[shift={(-0.3,0)}]
   \node[vertex] (x) at (0,0) {};
   
   \foreach \i/\a in {1/90,2/210,3/330} {
     \node[vertex, position=\a:0.5 from x] (u\i) {};
     \node[vertex, position=\a:1.2 from x] (v\i) {};
   
     \draw[very thick] (x) -- (u\i) -- (v\i);
   }
   \end{scope}
   
   \begin{scope}[shift={(5,0.3)}]
   \foreach \i in {1,2,3,4} {
     \node[vertex] (a\i) at (-3.5+1.4*\i,0.7) {};
     \node[vertex] (b\i) at (-3.5+1.4*\i,-0.7) {};
   }
   \foreach \i/\j in {1/1,1/2,1/3,1/4,2/2,2/3,2/4,3/3,3/4,4/4} {
     \draw[very thick] (a\i) -- (b\j);
   }
   
   \foreach \i/\j in {1/2,2/3,3/4} {
     \draw[very thick] (b\i) -- (b\j);
   }
   \foreach \i/\j in {1/3,2/4} {
     \draw[very thick] (b\i) .. controls +(0.7,-0.3) and +(-0.7,-0.3) .. (b\j);
   }
   \draw[very thick] (b1) .. controls +(0.85,-0.5) and +(-0.85,-0.5) .. (b4);
   \end{scope}
   
 \end{tikzpicture}
 \caption{A bull, a $1$-subdivided claw, and an example threshold graph.}\label{fig:gembull}
 \end{center}
\end{figure}

Fix a pattern graph $H$ without loops. We prove that \MWLHS can be solved:
 \begin{enumerate}[label=(R\arabic*),ref=(R\arabic*),nosep]
  \item\label{r:threshold}  in $\{P_5,F\}$-free graphs in polynomial time, whenever $F$ is a threshold graph;
  \item\label{r:bull} in $\{P_5,\textrm{bull}\}$-free graphs in polynomial time;
  \item\label{r:P5}   in $P_5$-free graphs in time $n^{\Oh(\omega(G))}$; and
  \item\label{r:S3}   in $\{P_6,\textrm{1-subdivided claw}\}$-free graphs in time $n^{\Oh(\omega(G)^3)}$.
 \end{enumerate}
Here, $n$ is the number of vertices of the input graph $G$ and $\omega(G)$ is the size of the maximum clique in~$G$.
Also, recall that a graph $G$ is a {\em{threshold graph}} if $V(G)$ can be partitioned into an independent set $A$ and a clique $B$ such that for each $a,a'\in A$, we have either $N(a)\supseteq N(a')$
or $N(a)\subseteq N(a')$. There is also a characterization via forbidden induced subgraphs: threshold graphs are exactly $\{2P_2,C_4,P_4\}$-free graphs, where $2P_2$ is an induced matching of size 2.
\cref{fig:gembull} depicts a bull, a $1$-subdivided claw, and an example threshold graph.

Further, we observe that by employing a simple branching strategy, 
an $n^{\Oh(\omega(G)^{\alpha})}$-time algorithm for \MWLHS in $\Ff$-free graphs can be used to give also a subexponential-time algorithm in this setting, with running time $n^{\Oh(n^{\alpha/(\alpha+1)})}$.
Thus, results~\ref{r:P5} and~\ref{r:S3} imply that for every fixed irreflexive $H$,
the \MWLHS problem can be solved in time $n^{\Oh(\sqrt{n})}$ in $P_5$-free graphs and in time $n^{\Oh(n^{3/4})}$ in  $\{P_6,\textrm{1-subdivided claw}\}$-free graphs.
This in particular applies to the {\sc{Odd Cycle Transversal}} problem.
We note here that Dabrowski et al.~\cite{DabrowskiFJPR19} proved that {\sc{Odd Cycle Transversal}} in $\{P_6,K_4\}$-free graphs is $\mathsf{NP}$-hard and does not admit a subexponential-time algorithm 
under the Exponential Time Hypothesis. Thus, it is unlikely that any of our algorithmic results --- the $n^{\Oh(\omega(G))}$-time algorithm and the $n^{\Oh(\sqrt{n})}$-time algorithm --- can be extended
from $P_5$-free graphs to $P_6$-free graphs.

All our algorithms work in a weighted setting, where instead of just maximizing the size of the domain of an $H$-coloring, we maximize its total {\em{revenue}}, 
where for each pair $(u,v)\in V(G)\times V(H)$ we have a prescribed revenue yielded by sending $u$ to $v$.
This setting allows encoding a broader range of coloring problems. For instance, list variants can be expressed by giving negative revenues for forbidden assignments (see e.g.~\cite{DBLP:journals/ejc/GutinHRY08,DBLP:journals/jcss/OkrasaR20}).
Also, our algorithms work in a slightly larger generality than stated above, see \cref{sec:corollaries}, \cref{sec:threshold}, and \cref{sec:bulls} for precise statements.

Finally, we investigate the possibility of extending our algorithmic results to pattern graphs with possible loops.
We show an example of a graph $H$ with loops, for which \MWLHS is $\mathsf{NP}$-hard and admits no subexponential-time algorithm under the ETH even in very restricted subclasses of $P_5$-free graphs, including
$\{P_5,\textrm{bull}\}$-free graphs. This shows that whether the pattern graph is allowed to have loops has a major impact on the complexity of the problem.

\paragraph*{Our techniques.} The key element of our approach is a branching procedure that, given an instance $(G,\wei)$ of \MWLHS, where $\wei$ is the revenue function, 
produces a relatively small set of instances $\Pi$ such that solving $(G,\wei)$ reduces to solving all the instances in $\Pi$. Moreover, every instance $(G',\wei')\in \Pi$ is simpler in the following sense:
either it is an instance of {\sc{Max Partial $H'$-Coloring}} for $H'$ being a proper induced subgraph of $H$ (hence it can be solved by induction on $|V(H)|$), or
for any connected graph $F$ on at least two vertices, $G'$ is $F$-free provided we assume $G$ is $\ante{F}$-free.
Here $\ante{F}$ is the graph obtained from $F$ by adding a universal vertex $y$ and a degree-$1$ vertex $x$ adjacent only to $y$.
In particular we have $\omega(G')<\omega(G)$, so applying the branching procedure exhaustively in a recursion scheme yields a recursion tree of depth bounded by $\omega(G)$. 
Now, for results~\ref{r:P5} and~\ref{r:S3} we respectively have $|\Pi|\leq n^{\Oh(1)}$ and $|\Pi|\leq n^{\Oh(\omega(G)^2)}$, giving bounds of $n^{\Oh(\omega(G))}$ and $n^{\Oh(\omega(G)^3)}$ on the total size of the recursion
tree and on the overall time complexity.

For result~\ref{r:threshold} we apply the branching procedure not exhaustively, but a constant number of times: if the original graph $G$ is $\{P_5,F\}$-free for some threshold graph $F$, it suffices to apply
the branching procedure $\Oh(|V(F)|)$ times to reduce the original instances to a set of edgeless instances, which can be solved trivially. As $\Oh(|V(F)|)=\Oh(1)$, this gives recursion tree of polynomial size, and hence
a polynomial-time complexity due always having $|\Pi|\leq n^{\Oh(1)}$ in this setting. For result~\ref{r:bull}, we show that two applications of the branching procedure reduce the input instance to a polynomial number
of instances that are $P_4$-free, which can be solved in polynomial time due to $P_4$-free graphs (also known as {\em{cographs}}) having cliquewidth at most $2$.
However, these applications are interleaved with a reduction to the case of {\em{prime graphs}} --- graphs with no non-trivial modules --- which we achieve using dynamic programming on the modular decomposition of the input graph.
This is in order to apply some results on the structure of prime $\textrm{bull}$-free graphs~\cite{ChudnovskyS08d,ChudnovskyS18}, so that $P_4$-freeness is achieved at the end.

Let us briefly discuss the key branching procedure. The first step is finding a useful dominating structure that we call a {\em{monitor}}: a subset of vertices $M$ of a connected graph $G$ is a monitor if for every connected
component $C$ of $G-M$, there is a vertex in $M$ that is complete to $C$. We prove that in a connected $P_6$-free graph there is always a monitor that is the closed neighborhood of a set of at most three vertices. After finding
such a monitor $N[X]$ for $|X|\leq 3$, we perform a structural analysis of the graph centered around the set $X$. This analysis shows that there exists a subset of $\Oh(|V(H)|)$ vertices such that after guessing this subset and 
the $H$-coloring on it, the instance can be partitioned into several separate subinstances, each of which has a strictly smaller clique number. This structural analysis, and in particular the way the separation of subinstances
is achieved, is inspired by the polynomial-time algorithm of Ho{\`a}ng et al.~\cite{kcolp5free} for {\sc{$k$-Coloring}} in $P_5$-free graphs.

\paragraph*{Other related work.} We remark that very recently and independently of us, Brettell et al.~\cite{BrettellHP20} proved that for every fixed $s,t\in \N$, the class of $\{K_t,sK_1+P_5\}$-free graphs has bounded mim-width.
Here, {\em{mim-width}} is a graph parameter that is less restrictive than cliquewidth, but the important aspect is that a wide range of vertex-partitioning problems, including the \MWLHS problem considered in this work, can be solved
in polynomial time on every class of graphs where the mim-width is universally bounded and a corresponding decomposition can be computed efficiently. The result of Brettell et al. thus shows that in $P_5$-free graphs,
the mim-width is bounded by a function of the clique number. This gives an $n^{f(\omega(G))}$-time algorithm for \MWLHS in $P_5$-free graphs (for fixed $H$), for some function~$f$. 
However, the proof presented in~\cite{BrettellHP20} gives only an exponential upper bound on the function $f$, which in particular does not imply the existence of a subexponential-time algorithm.
To compare, our reasoning leads to an $n^{\Oh(\omega(G))}$-time algorithm and a subexponential-time algorithm with complexity $n^{\Oh(\sqrt{n})}$.

We remark that the techniques used by Brettell et al.~\cite{BrettellHP20} also rely on revisiting the approach of Ho{\`a}ng et al.~\cite{kcolp5free}, 
and they similarly observe that this approach can be used to apply induction based on the clique number of the graph.

\paragraph*{Organization.} After setting up notation and basic definition in \cref{sec:prelims} and proving an auxiliary combinatorial result about $P_6$-free graphs in \cref{sec:monitors},
we provide the key technical lemma (\cref{lem:main-branching-full}) in \cref{sec:branching}. 
This lemma captures a single branching step of our algorithms. In \cref{sec:corollaries} we derive results~\ref{r:P5} and~\ref{r:S3}.
\cref{sec:threshold} and \cref{sec:bulls} are devoted to the proofs of results~\ref{r:threshold} and~\ref{r:bull}, respectively.
In \cref{sec:hardness} we show that allowing loops in $H$ may result in an $\mathsf{NP}$-hard problem even in restricted subclasses of $P_5$-free graphs.
We conclude in \cref{sec:conclusions} by discussing directions of further research.

\section{Preliminaries}\label{sec:prelims}

\paragraph*{Graphs.}
For a graph $G$, the vertex and edge sets of $G$ are denoted by $V(G)$ and $E(G)$, respectively.
The {\em{open neighborhood}} of a vertex $u$ is the set $N_G(u)\coloneqq \{v\colon uv\in E(G)\}$, while the {\em{closed neighborhood}} is $N_G[u]\coloneqq N_G(u)\cup \{u\}$.
This notation is extended to sets of vertices: for $X\subseteq V(G)$, we set $N_G[X]\coloneqq \bigcup_{u\in X} N_G[u]$ and $N_G(X)\coloneqq N_G[X]\setminus X$.
We may omit the subscript if the graph $G$ is clear from the context. By $C_t$, $P_t$, and $K_t$ we respectively denote the cycle, the path, and the complete graph on $t$ vertices.

The {\em{clique number}} $\omega(G)$ is the size of the largest clique in a graph $G$.
A clique $K$ in $G$ is {\em{maximal}} if no proper superset of $K$ is a clique.

For $s,t\in \N$, the {\em{Ramsey number}} of $s$ and $t$ is the smallest integer $k$ such that every graph on $k$ vertices contains either a clique of size $s$ or an independent set of size $t$.
It is well-known that the Ramsey number of $s$ and $t$ is bounded from above by $\binom{s+t-2}{s-1}$, hence we will denote $\Ramsey(s,t)\coloneqq \binom{s+t-2}{s-1}$.

For a graph $G$ and $A\subseteq V(G)$, by $G[A]$ we denote the subgraph of $G$ induced by $A$. We write $G-A\coloneqq G[V(G)\setminus A]$.
We say that $F$ is an {\em{induced subgraph}} of $G$ if there is $A\subseteq V(G)$ such that $G[A]$ is isomorphic to $F$; this containment is {\em{proper}} if in addition $A\neq V(G)$.
For a family of graphs $\Ff$, a graph $G$ is {\em{$\Ff$-free}} if $G$ does not contain any induced subgraph from $\Ff$.
If $\Ff=\{H\}$, then we may speak about {\em{$H$-free}} graphs as well.

If $G$ is a graph and $A\subseteq V(G)$ is a subset of vertices, then a vertex $u\notin A$ is {\em{complete}} to $A$ if $u$ is adjacent to all the vertices of $A$, 
and $u$ is {\em{anti-complete}} to $A$ if $u$ has no neighbors in $A$.
We will use the following simple claim several times.

\begin{lemma}\label{lem:mixed-P3}
 Suppose $G$ is a graph, $A$ is a subset of its vertices such that $G[A]$ is connected, and $u\notin A$ is a vertex that is neither complete nor anti-complete to $A$ in $G$.
 Then there are vertices $a,b\in X$ such that $u-a-b$ is an induced $P_3$ in $G$.
\end{lemma}
\begin{proof}
Since $u$ is neither complete nor anticomplete to $A$, both the
sets $A\cap N(u)$ and $A\setminus N(u)$ are non-empty. As $A$ is connected, 
there exist $a \in A\cap N(u)$ and $b \in A\setminus N(u)$ such that 
$a$ and $b$ are adjacent. Now $u-a-b$ is the desired induced $P_3$.
\end{proof}

For a graph $F$, by $\univ{F}$ we denote the graph obtained from $F$ by adding a {\em{universal vertex}}: a vertex adjacent to all the other vertices.
Similarly, by $\ante{F}$ we denote the graph obtained from $F$ by adding first an isolated vertex, say $x$, and then a universal vertex, say $y$.
Note that thus $y$ is adjacent to all the other vertices of $\ante{F}$, while $x$ is adjacent only to $y$.

\paragraph*{$H$-colorings.}
For graphs $H$ and $G$, a function $\phi\colon V(G)\to V(H)$ is a {\em{homomorphism}} from $G$ to $H$ if for every $uv\in E(G)$, we also have $\phi(u)\phi(v)\in E(H)$.
Note that a homomorphism from $G$ to the complete graph $K_t$ is nothing else than a proper coloring of $G$ with $t$ colors.
Therefore, a homomorphism from $G$ to $H$ will be also called an {\em{$H$-coloring}} of $G$, and we will refer to vertices of $H$ as colors.
Note that we will always assume that $H$ is a simple graph without loops, so no two adjacent vertices of $G$ can be mapped by a homomorphism to the same vertex of $H$.
To stress this, we will call such $H$ an {\em{irreflexive pattern graph}}.

A {\em{partial homomorphism}} from $G$ to $H$, or a {\em{partial $H$-coloring}} of $G$,
is a partial function $\phi\colon V(G)\partto V(H)$ that is a homomorphism from $G[\dom \phi]$ to $H$, where $\dom \phi$ denotes the domain of $\phi$.

Suppose that with graphs $G$ and $H$ we associate a {\em{revenue function}} $\wei\colon V(G)\times V(H)\to \R$.
Then the {\em{revenue}} of a partial $H$-coloring $\phi$ is defined as
$$\wei(\phi)\coloneqq \sum_{u\in \dom \phi} \wei(u,\phi(u)).$$
In other words, for $u\in V(G)$ and $v\in V(H)$, $\wei(u,v)$ denotes the revenue yielded by assigning $\phi(u)\coloneqq v$.

We now define the main problem studied in this work. In the following, we consider the graph $H$ fixed.

\defproblem{\MWLHS}{Graph $G$ and a revenue function $\wei\colon V(G)\times V(H)\to \R$}
{A partial $H$-coloring $\phi$ of $G$ that maximizes $\wei(\phi)$}

An {\em{instance}} of the \MWLHS problem is a pair $(G,\wei)$ as above.
A {\em{solution}} to an instance $(G,\wei)$ is a partial $H$-coloring of $G$, and it is {\em{optimum}} if it maximizes $\wei(\phi)$ among solutions.
By $\OPT(G,\wei)$ we denote the maximum possible revenue of a solution to the instance $(G,\wei)$.

Let us note one aspect that will be used later on.
Observe that in revenue functions we allow negative revenues for some assignments.
However, if we are interested in maximizing the total revenue, there is no point in using such assignments: 
if $u\in \dom \phi$ and $\wei(u,\phi(u))<0$, then just removing $u$ from the domain of $\phi$ increases the revenue.
Thus, optimal solutions never use assignments with negative revenues.
Note that this feature can be used to model list versions of partial coloring problems.

\section{Monitors in $P_6$-free graphs}\label{sec:monitors}

In this section we prove an auxiliary result about finding useful separators in $P_6$-free graphs. The desired property is expressed in the following definition.

\begin{definition}
 Let $G$ be a connected graph. A subset of vertices $M\subseteq V(G)$ is a {\em{monitor}} in $G$ if 
 for every connected component $C$ of $G-M$, there exists a vertex $w\in M$ that is complete to $C$.
\end{definition}

Let us note the following property of monitors.

\begin{lemma}\label{lem:monitor-omega}
 If $M$ is a monitor in a connected graph $G$, then every maximal clique in $G$ intersects $M$.
 In particular, $\omega(G-M)<\omega(G)$.
\end{lemma}
\begin{proof}
 If $K$ is a clique in $G-M$, then $K$ has to be entirely contained in some connected component $C$ of~$G-M$.
 Since $M$ is a monitor, there exists $w\in M$ that is complete to $C$. Then $K\cup \{w\}$ is also a clique in $G$, hence $K$ cannot be a maximal clique in $G$.
\end{proof}

We now prove that in $P_6$-free graphs we can always find easily describable monitors.

\begin{lemma}\label{lem:hitting-cliques}
 Let $G$ be a connected $P_6$-free graph. Then for every $u\in V(G)$ there exists a subset of vertices $X$ such that $u\in X$, $|X|\leq 3$, $G[X]$ is a path whose one endpoint is $u$, and $N_G[X]$ is a monitor in $G$.
\end{lemma}

\cref{lem:hitting-cliques} follows immediately from the following statement applied for $t=6$.

\begin{lemma}\label{lem:monitor-technical}
 Let $t\in \{4,5,6\}$, $G$ be a connected $P_6$-free graph, and $u\in V(G)$ be a vertex such that in $G$ there is no induced $P_t$ with $u$ being one of the endpoints.
 Then there exists a subset $X$ of vertices such that $u\in X$, $|X|\leq t-3$, $G[X]$ is a path whose one endpoint is $u$, and $N_G[X]$ is a monitor in $G$.
\end{lemma}
\begin{proof}
 We proceed by induction on $t$. The base case for $t=4$ will be proved directly within the analysis.
 
 In the following, by {\em{slabs}} we mean connected components of the graph $G-N_G[u]$.
 We shall say that a vertex $w\in N_G(u)$ is {\em{mixed}} on a slab $C$ if $w$ is neither complete nor anti-complete to $C$.
 A slab $C$ is {\em{simple}} if there exists a vertex $w\in N_G(u)$ that is complete to $C$, and {\em{difficult}} otherwise.
 
 Note that since $G$ is connected, for every difficult slab $D$ there exists some vertex $w \in N_G(u)$ that is mixed on~$D$.
 Then, by \cref{lem:mixed-P3}, we can find vertices $a,b\in D$ such that $u-w-a-b$ is an induced $P_4$ in~$G$.
 If $t=4$ then no such induced $P_4$ can exists, so we infer that in this case there are no difficult slabs.
 Then $N_G[u]$ is a monitor, so we may set $X\coloneqq \{u\}$.
 This proves the claim for $t=4$; from now on we assume that $t\geq 5$. 
 
 Let us choose a vertex $v\in N_G(u)$ that maximizes the number of difficult slabs on which $v$ is mixed.
 Suppose there is a difficult slab $D'$ such that $v$ is anti-complete to $D'$.
 As we argued, there exists a vertex $v'\in N_G(u)$ such that $v'$ is mixed on $D'$; clearly $v'\neq v$.
 By the choice of $v$, there exists a difficult slab $D$ such that $v$ is mixed on $D$ and $v'$ is anti-complete to $D$.
 By applying \cref{lem:mixed-P3} twice, we find vertices $a,b\in D$ and $a',b'\in D'$ such that $v-a-b$ and $v'-a'-b'$ are induced $P_3$s in $G$.
 Now, if $v$ and $v'$ were adjacent, then $a-b-v-v'-a'-b'$ would be an induced $P_6$ in $G$, a contradiction.
 Otherwise $a-b-v-u-v'-a'-b'$ is an induced $P_7$ in $G$, again a contradiction (see \cref{fig:vismixed}).
 
 \begin{figure}[htb]
 \begin{center}
 \includegraphics[scale=1.0,page=1]{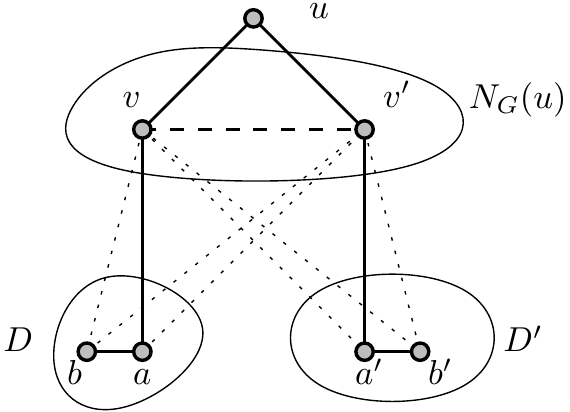}
 \end{center}
 \caption{The graph $G$ in the proof of \cref{lem:monitor-technical} when $v$ anti-complete to some difficult slab $D'$. Dotted lines show non-edges. The edge $vv'$ might be present.}\label{fig:vismixed}
 \end{figure}
 
 We conclude that $v$ is mixed on every difficult slab. Let 
 $$A\coloneqq \{v\}\cup \bigcup_{D\colon \textrm{difficult slab}} V(D).$$
 Then $G[A]$ is connected and $P_6$-free. Moreover, in $G[A]$ there is no $P_{t-1}$ with one endpoint being $v$, because otherwise we would be able to extend such an induced $P_{t-1}$ using $u$, and thus 
 obtain an induced $P_t$ in $G$ with one endpoint being $u$.
 Consequently, by induction we find a subset $Y\subseteq A$ such that $|Y|\leq t-4$, $G[Y]$ is a path with one of the endpoints being $v$, and $N_{G[A]}[Y]$ is a monitor in $G[A]$.
 Let $X\coloneqq Y\cup \{u\}$. Then $|X|\leq t-3$ and $G[X]$ is a path with $u$ being one of the endpoints.
 
 We verify that $N_G[X]$ is a monitor in $G$. Consider any connected component $C$ of $G-N_G[X]$. As $N_G[X]\supseteq N_G[u]$, $C$ is contained in some slab $D$.
 If $D$ is simple, then by definition there exists a vertex $w\in N_G[u]\subseteq N_G[X]$ that is complete to $D$, hence also complete to $C$.
 Otherwise $D$ is difficult, hence $C$ is a connected component of $G[A]-N_{G[A]}[Y]$. Since $N_{G[A]}[Y]$ is a monitor in $G[A]$, there exists a vertex $w\in N_{G[A]}[Y]\subseteq N_G[X]$ that is complete to $C$.
 This completes the proof.
\end{proof}

We remark that no statement analogous to \cref{lem:hitting-cliques} may hold for $P_7$-free graphs, even if from $X$ we only require that $N_G[X]$ intersects all the maximum-size cliques in $G$ 
(which is implied by the property of being a monitor, see \cref{lem:monitor-omega}).
Consider the following example. Let $G$ be a graph obtained from the union of $n+1$ complete graphs $K^{(0)},\ldots,K^{(n)}$, each on $n$ vertices, 
by making one vertex from each of the graphs $K^{(1)},\ldots,K^{(n)}$ adjacent to a different vertex of $K^{(0)}$.
Then $G$ is $P_7$-free, but the minimum size of a set $X\subseteq V(G)$ such that $N_G[X]$ intersects all maximum-size cliques in $G$ is $n$.

\section{Branching}\label{sec:branching}

We now present the core branching step that will be used by all our algorithms. This part is inspired by the approach of Ho\`ang et al.~\cite{kcolp5free}.
We will rely on the following two graph families; see \cref{fig:LsSt}.
For $t \in \mathbb{N}$, the graph $S_t$ is obtained from the star $K_{1,t}$ by subdividing every edge once.
Then $L_1 \coloneqq P_3$ and for $t \geq 2$ the graph $L_t$ is obtained from $S_t$ by making all the leaves of $S_t$ pairwise adjacent.

\begin{figure}[htb]
 \begin{center}
 \begin{tikzpicture}
  
   \tikzstyle{vertex}=[circle,thick,draw=black,fill=gray!20,minimum size=0.2cm,inner sep=0pt]

   \begin{scope}[shift={(-3,0)}]
   \node[vertex] (x) at (0,0) {};
   
   \foreach \i/\a in {1/210,2/250,3/290,4/330} {
     \node[vertex, position=\a:1 from x] (u\i) {};
     \node[vertex, position=270:1 from u\i] (w\i) {};
     \draw[very thick] (x) -- (u\i) -- (w\i);
   }
   \end{scope}
   
   \begin{scope}[shift={(3,0)}]
   \node[vertex] (x) at (0,0) {};
   
   \foreach \i/\a in {1/210,2/250,3/290,4/330} {
     \node[vertex, position=\a:1 from x] (u\i) {};
     \node[vertex, position=270:1 from u\i] (w\i) {};
     \draw[very thick] (x) -- (u\i) -- (w\i);
   }
   \foreach \i/\j in {1/2,1/3,1/4,2/3,2/4,3/4} {
     \draw[very thick] (w\i) -- (w\j);
   }
   
   \end{scope}
   
 \end{tikzpicture}
 \caption{Graphs $S_4$ and $L_4$.}\label{fig:LsSt}
 \end{center}
\end{figure}
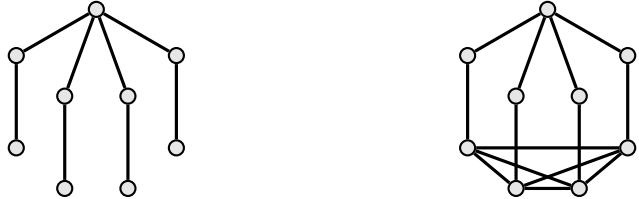


\begin{lemma}\label{lem:main-branching-full}
 Let $H$ be a fixed irreflexive pattern graph.
 Suppose we are given integers $s,t$ and an instance $(G,\wei)$ of \MWLHS such that $G$ is connected and $\{P_6, L_s, S_t\}$-free, and the range of $\wei(\cdot)$ contains at least one positive value.
 Denoting $n\coloneqq |V(G)|$, one can in time $n^{\Oh(\Ramsey(s,t))}$ compute a set $\Pi$ of size $n^{\Oh(\Ramsey(s,t))}$ such that the following conditions hold:
\begin{enumerate}[label=(B\arabic*),ref=(B\arabic*),nosep]
 \item\label{b:basic}
       Each element of $\Pi$ is a pair $((G_1,\wei_1),(G_2,\wei_2))$, where $G_1,G_2$ are $\{P_6, L_s, S_t\}$-free subgraphs of~$G$ satisfying $V(G)=V(G_1)\uplus V(G_2)$. 
       Further, $(G_2,\wei_2)$ is an instance of \MWLHS, and $(G_1,\wei_1)$ is an
       instance of \textsc{Max Partial $H'$-Coloring}, where $H'$ is some proper induced subgraph of $H$ (which may be different for different elements of $\Pi$).
 \item\label{b:progress}
       For each $((G_1,\wei_1),(G_2,\wei_2))\in \Pi$ and every connected graph $F$ on at least two vertices, if $G_1$ contains an induced $F$, then $G$ contains an induced $\univ{F}$.
       Moreover, if $G_2$ contains an induced $F$, then $G$ contains an induced $\ante{F}$.
 \item\label{b:soundness}
       We have $$\OPT(G,\wei) = \max\ \{\,\OPT(G_1,\wei_1)+\OPT(G_2,\wei_2)\ \colon\ ((G_1,\wei_1),(G_2,\wei_2)) \in \Pi\,\}.$$
       Moreover, for any pair $((G_1,\wei_1),(G_2,\wei_2))\in \Pi$ for which this maximum is reached, and for every pair of optimum solutions $\phi_1$ and $\phi_2$ to $(G_1,\wei_1)$ and $(G_2,\wei_2)$, respectively, 
       the function $\phi\coloneqq \phi_1\cup \phi_2$ is an optimum solution to $(G,\wei)$ with $\wei(\phi)=\wei_1(\phi_1)+\wei_2(\phi_2)$.
\end{enumerate}
\end{lemma}

The remainder of this section is devoted to the proof of \cref{lem:main-branching-full}. We fix the irreflexive pattern graph~$H$ and consider an input instance $(G,\wei)$.
We find it more didactic to first perform an analysis of $(G,\wei)$, and only provide the algorithm at the end. Thus, the correctness will be clear from 
the previous observations.

Let $$T\coloneqq \{\,(x,y)\in V(G)\times V(H)\ \colon\ \wei(x,y)>0\,\}.$$
By assumption $T$ is nonempty, hence $\OPT(G,\wei)>0$ and every optimum solution $\phi$ to $(G,\wei)$ has a nonempty domain: it sets $\phi(x)=y$ for some $(x,y)\in T$.
Consequently, the final set $\Pi$ will be obtained by taking the union of sets $\Pi^{x,y}$ for $(x,y)\in T$: when constructing $\Pi^{x,y}$ our goal is to capture all solutions satisfying $\phi(x)=y$.
We now focus on constructing $\Pi^{x,y}$, hence we assume that we fix a pair $(x,y)\in T$.

\medskip

Since $G$ is connected, by \cref{lem:hitting-cliques} there exists $X\subseteq V(G)$ such that $x\in X$, $|X|\leq 3$, $G[X]$ is a path with $x$ being one of the endpoints, and $N[X]$ is a monitor in~$G$.
Note that such $X$ can be found in polynomial time by checking all subsets of $V(G)\setminus \{x\}$ of size at most $2$.
In case $|X|<3$, we may add arbitrary to $X$ so that $|X|=3$ and $G[X]$ remains connected; note that this does not spoil the property that $G[X]$ is a monitor.
We may also enumerate the vertices of $X$ as $\{x_1,x_2,x_3\}$ so that $x=x_1$ and for each $i\in \{2,3\}$ there exists $i'<i$ such that $x_i$ and $x_{i'}$ are adjacent.

We partition $V(G)\setminus X$ into $A_1,A_2,A_3,A_4$ as follows:
$$A_1\coloneqq N(x_1)\setminus X,\quad A_2\coloneqq N(x_2)\setminus (X\cup A_1),\quad A_3\coloneqq N(x_3)\setminus (X\cup A_1\cup A_2),\quad A_4\coloneqq V(G)\setminus N[X].$$ 
Note that $\{A_1,A_2,A_3\}$ is a partition of $N(X)$ (see \cref{fig:partitionAi}).
For $i\in \{1,2,3\}$, denote $A_{>i}\coloneqq \bigcup_{j=i+1}^4 A_j$ and observe that $x_i$ is complete to $A_i$ and anti-complete to $A_{>i}$.
Moreover, we have the following.

\begin{claim}\label{cl:Ai-good}
 Let $F$ be a connected graph. 
 If $G[A_1]$ contains an induced $F$, then $G$ contains an induced $\univ{F}$.
 If $G[A_i]$ contains an induced $F$ for any $i\in \{2,3,4\}$, then $G$ contains an induced $\ante{F}$.
\end{claim}
\begin{clproof}
 For the first assertion observe that if $B\subseteq A_1$ induces $F$ in $G$, then $B\cup \{x_1\}$ induces $\univ{F}$ in $G$.
 For the second assertion, consider first the case when $i\in \{2,3\}$. As we argued, there is $i'<i$ such that $x_{i'}$ and $x_i$ are adjacent.
 Then if $B\subseteq A_i$ induces $F$ in $G$, then $B\cup \{x_{i'},x_i\}$ induces $\ante{F}$ in $G$.
 
 We are left with justifying the second assertion for $i=4$.
 Suppose $B\subseteq A_4$ induces $F$ in $G$. Since $F$ is connected, $B$ is entirely contained in one connected component $C$ of $G[A_4]$.
 As $N[X]$ is a monitor in $G$, there exists a vertex $w\in N[X]$ that is complete to $C$. As $w\in N[X]$, some $x_{i'}\in X$ is adjacent to $w$.
 We now find that $B\cup \{w,x_{i'}\}$ induces $\ante{F}$ in $G$.
\end{clproof}

 \begin{figure}[htb]
 \begin{center}
 \includegraphics[scale=1.0,page=2]{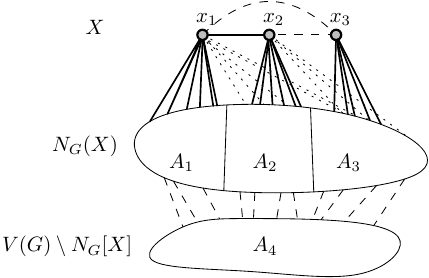}
 \end{center}
 \caption{The partition on $V(G)$ in the proof of \cref{lem:main-branching-full}. Solid and dotted lines respectively indicate that a vertex is complete or anticomplete to a set. 
          Dashed edges might, but do not have to exist.}\label{fig:partitionAi}
 \end{figure}

The next claim contains the core combinatorial observation of the proof.

\begin{claim}\label{cl:small-guess}
 Let $\phi$ be a solution to the instance $(G,\wei)$. Then for every $i\in \{1,2,3\}$ and $v\in V(H)$, there exists a set $S\subseteq A_i$ such that:
 \begin{itemize}[nosep]
  \item $|S|<\Ramsey(s,t)$;
  \item $S\subseteq A_i\cap \phi^{-1}(v)$; and
  \item every vertex $u\in A_{>i}$ that has a neighbor in $A_i\cap \phi^{-1}(v)$, also has a neighbor in $S$.
 \end{itemize}
\end{claim}
\begin{clproof}
 Let $S$ be the smallest set satisfying the second and the third condition, it exists, as these conditions are satisfied by $A_i \cap \phi^{-1}(v)$ .
 Note that since $H$ is irreflexive, it follows that $\phi^{-1}(v)$ is an independent set in $G$, hence $S$ is independent as well.
 
 Suppose for contradiction that $|S|\geq \Ramsey(s,t)$.
 By minimality, for every $u\in S$ there exists $u'\in A_{>i}$ such that $u$ is the only neighbor of $u'$ in $S$.
 Let $S'\coloneqq \{u'\colon u\in S\}$ (see \cref{fig:SSprime}).
 Since $|S'|\geq \Ramsey(s,t)$, in $G[S']$ we can either find a clique $K'$ of size $s$ or an independent set $I'$ of size $t$; denote $K\coloneqq \{u\colon u'\in K'\}$ and $I\coloneqq \{u\colon u'\in I'\}$.
 In the former case, we find that $\{x_i\}\cup K\cup K'$ induces the graph $L_s$ in $G$, a contradiction.
 Similarly, in the latter case we have that $\{x_i\}\cup I\cup I'$ induces $S_t$ in $G$, again a contradiction.
 This completes the proof of the claim.
\end{clproof}

 \begin{figure}[htb]
 \begin{center}
 \includegraphics[scale=1.0,page=3]{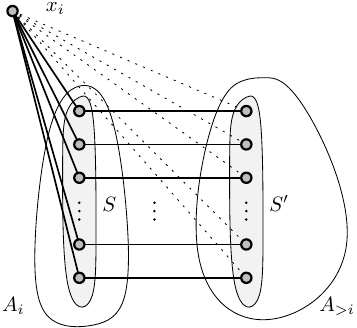}
 \end{center}
 \caption{Sets $S$ and $S'$ in the proof of \cref{cl:small-guess}.}\label{fig:SSprime}
 \end{figure}

\cref{cl:small-guess} suggests the following notion.
A {\em{guess}} is a function $R\colon V(H)\to 2^{N[X]}$ satisfying the following: 
\begin{itemize}[nosep]
 \item for each $v\in V(H)$, $R(v)$ is a subset of $N[X]$ such that $|R(v)\cap A_i|<\Ramsey(s,t)$ for all $i\in \{1,2,3\}$;
 \item sets $R(v)$ are pairwise disjoint for different $v\in V(H)$; and
 \item $x\in R(y)$.
\end{itemize}
Let $\Rr^{x,y}$ be the family of all possible guesses. Note that we add the pair $(x,y)$ in the superscript to signify that the definition of $\Rr^{x,y}$ depends on $(x,y)$.

\begin{claim}\label{cl:branching-poly}
 We have that $|\Rr^{x,y}|\leq n^{\Oh(\Ramsey(s,t))}$ and $\Rr^{x,y}$ can be enumerated in time $n^{\Oh(\Ramsey(s,t))}$.
\end{claim}
\begin{clproof}
 For each $v\in V(H)$, the number of choices for $R(v)$ in a guess $R$ is bounded by $2^3 \cdot n^{3\cdot \Ramsey(s,t)}$:
 the first factor corresponds to the choice of $R(v)\cap X$, while the second factor bounds the number of 
 choices of $R(v)\cap A_i$ for $i\in \{1,2,3\}$. Since the guess $R$ is determined by choosing $R(v)$ for each $v\in V(H)$ and $|V(H)|$ is considered a constant, the number of different guesses is bounded by
 $\left(2^3\cdot n^{3\cdot \Ramsey(s,t)}\right)^{|V(H)|}=n^{\Oh(\Ramsey(s,t))}$. Clearly, they can be also enumerated in time $n^{\Oh(\Ramsey(s,t))}$.
\end{clproof}

Now, we say that a guess $R$ is {\em{compatible}} with a solution $\phi$ to $(G,\wei)$ if the following conditions hold for every $v\in V(H)$:
\begin{enumerate}[label=(C\arabic*),ref=(C\arabic*),nosep]
 \item\label{c:contained} $R(v)\subseteq \phi^{-1}(v)$;
 \item\label{c:X}         $R(v)\cap X = \phi^{-1}(v)\cap X$; and
 \item\label{c:witness}   for all $i\in \{1,2,3\}$ and $u\in A_{>i}$, if $u$ has a neighbor in $\phi^{-1}(v)\cap A_i$, then $u$ also has a neighbor in~$R(v)\cap A_i$. 
\end{enumerate}
The following statement follows immediately from \cref{cl:small-guess}.

\begin{claim}\label{cl:branching-correct}
 For every solution $\phi$ to the instance $(G,\wei)$ which satisfies $\phi(x)=y$, there exists a guess $R\in \Rr^{x,y}$ that is compatible with $\phi$.
\end{claim}

Let us consider a guess $R\in \Rr^{x,y}$.
We define a set $B^R\subseteq V(G)\times V(H)$ of {\em{disallowed pairs}} for $R$ as follows. We include a pair $(u,v)\in V(G)\times V(H)$ in $B^R$ if any of the following four conditions holds:
\begin{enumerate}[label=(D\arabic*),ref=(D\arabic*),nosep]
 \item\label{d:X}        $u\in X$ and $u\notin R(v)$;
 \item\label{d:R}        $u\in R(v')$ for some $v'\in V(H)$ that is different from $v$;                                      
 \item\label{d:nei}      $u$ has a neighbor in $G$ that belongs to $R(v')$ for some $v'\in V(H)$ such that $vv'\notin E(H)$; or
 \item\label{d:no-nei}   $u\in A_i\setminus R(v)$ for some $i\in \{1,2,3\}$ and there exists $u'\in A_{>i}$ such that $uu'\in E(G)$ and $N_G(u')\cap A_i\cap R(v)=\emptyset$.  
\end{enumerate}
Intuitively, $B^R$ contains assignments that contradict the supposition that $R$ is compatible with a considered solution.
The fact that $x=x_1$ is complete to $A_1$ and the assumption $x\in R(y)$ directly yield the following.

\begin{claim}\label{cl:A1-simpler}
 For all $u\in A_1$ and $R\in \Rr^{x,y}$, we have $(u,y)\in B^R$.
\end{claim}

Based on $B^R$, we define a new revenue functions $\wei^R\colon V(G)\times V(H)\to \R$ as follows:
$$\wei^R(u,v)=\begin{cases} -1 &\quad  \textrm{if }(u,v)\in B^R;\\ \wei(u,v) &\quad  \textrm{otherwise.}\end{cases}$$
The intuition is that if a pair $(u,v)$ is disallowed by $R$, then we model this in $\wei^R$ by assigning negative revenue to the corresponding assignment. 
This forbids optimum solutions to use this assignment.

We now define a subgraph $G^{x,y}$ of $G$ as follows:
$$V(G^{x,y})\coloneqq V(G)\qquad\textrm{and}\qquad E(G^{x,y})\coloneqq \{\, uv\in E(G)\ \colon\ u,v\in A_i\textrm{ for some }i\in \{1,2,3,4\}\,\}$$
In other words, $G^{x,y}$ is obtained from $G$ by removing all edges except those whose both endpoints belong to the same set $A_i$, for some $i\in \{1,2,3,4\}$.

For every guess $R\in \Rr^{x,y}$, we may consider a new instance $(G^{x,y},\wei^R)$ of \MWLHS.
In the following two claims we establish the relationship between solutions to the instance $(G,\wei)$ and solutions to instances $(G^{x,y},\wei^R)$.
The proofs essentially boil down to a verification that all the previous definitions work as expected. 
In particular, the key point is that the modification of revenues applied when constructing $\wei^R$ implies automatic satisfaction of all the constraints associated with edges that were present in $G$,
but got removed in $G^{x,y}$.

\begin{claim}\label{cl:not-larger}
 For every guess $R\in \Rr^{x,y}$, every optimum solution $\phi$ to the instance $(G^{x,y},\wei^R)$ is also a solution to the instance $(G,\wei)$, and moreover $\wei^R(\phi)=\wei(\phi)$.
\end{claim}
\begin{clproof}
 Recall that $\phi$ is a solution to $(G,\wei)$ if and only if $\phi$ is a partial $H$-coloring of $G$. 
 Hence, we need to prove that for every $uu'\in E(G)$ with $u,u'\in \dom \phi$, we have $\phi(u)\phi(u')\in E(H)$.
 Denote $v\coloneqq\phi(u)$ and $v'\coloneqq\phi(u')$ and suppose for contradiction that $vv'\notin E(H)$.
 Since $\phi$ is an optimum solution to $(G^{x,y},\wei^R)$, we have $\wei^R(u,v)\geq 0$, which implies that $(u,v)\notin B^R$. Similarly $(u',v')\notin B^R$.
 We now consider cases depending on the alignment of $u$ and $u'$ in $G$.
 
 If $u,u'\in A_i$ for some $i\in \{1,2,3,4\}$ then $uu'\in E(G^{x,y})$, so the supposition $vv'\notin E(H)$ would contradict the assumption that $\phi$ is a solution to $(G^{x,y},\wei^R)$.
 
 Suppose $u\in A_i$ and $u'\in A_j$ for $i,j\in \{1,2,3,4\}$, $i\neq j$; by symmetry, assume $i<j$.
 As $vv'\notin E(H)$, we infer that $u'$ does not have any neighbors in $R(v)$ in $G$, for otherwise we would have $(u',v')\in B^R$ by~\ref{d:nei}.
 As $uu'\in E(G)$, $u\in A_i$, and $u'\in A_{>i}$, this implies that $(u,v)\in B^R$ by~\ref{d:no-nei}, a contradiction.
 
 Finally, suppose that $\{u,u'\}\cap X\neq \emptyset$, say $u\in X$.
 Since $(u,v)\notin B^R$, by~\ref{d:X} we infer that $u\in R(v)$.
 Then, by~\ref{d:nei}, $vv'\notin E(H)$ and $uu'\in E(G)$ together imply that $(u',v')\in B^R$, a contradiction.
 
 This finishes the proof that $\phi$ is a solution to $(G,\wei)$.
 To see that $\wei^R(\phi)=\wei(\phi)$ note that $\phi$, being an optimum solution to $(G^{x,y},\wei^R)$, does not use any assignments with negative revenues in $\wei^R$, 
 while $\wei(u,v)=\wei^R(u,v)$ for all $(u,v)$ satisfying $\wei^R(u,v)\geq 0$.
\end{clproof}

\begin{claim}\label{cl:not-smaller}
 If $\phi$ is a solution to $(G,\wei)$ that is compatible with a guess $R\in \Rr^{x,y}$, then $\phi$ is also a solution to $(G^{x,y},\wei^R)$ and $\wei^R(\phi)=\wei(\phi)$.
\end{claim}
\begin{clproof}
 As $\phi$ is a solution to $(G,\wei)$, it is a partial $H$-coloring of $G$.
 Since $G^{x,y}$ is a subgraph of $G$ with $V(G^{x,y})=V(G)$, $\phi$ is also a partial $H$-coloring of $G^{x,y}$.
 Hence $\phi$ is a solution to $(G^{x,y},\wei^R)$.
 
 To prove that $\wei^R(\phi)=\wei(\phi)$ it suffices to show that $(u,\phi(u))\notin B^R$ for every $u\in \dom \phi$, since functions $\wei^R$ and $\wei$ differ only on the pairs from $B^R$.
 Suppose otherwise, and consider cases depending on the reason for including $(u,\phi(u))$ in $B^R$. Denote $v\coloneqq \phi(u)$.
 
 First, suppose $u\in X$ and $u\notin R(v)$. By~\ref{c:X} we have $u\notin R(v)\cap X=\phi^{-1}(v)\cap X\ni u$, a contradiction.
 
 Second, suppose $u\in R(v')$ for some $v'\neq v$. By~\ref{c:contained} we have $v=\phi(u)=v'$, again a contradiction.
 
 Third, suppose that $u$ has a neighbor $u'$ in $G$ such that $u'\in R(v')$ for some $v'\in V(H)$ satisfying $vv'\notin E(H)$.
 By~\ref{c:contained}, we have $u'\in \dom \phi$ and $\phi(u')=v'$. But then $\phi(u)\phi(u')=vv'\notin E(H)$ even though $uu'\in E(G)$, a contradiction with the assumption that $\phi$ is a partial $H$-coloring of $G$.
 
 Fourth, suppose that $u\in A_i\setminus R(v)$ for some $i\in \{1,2,3\}$ and there exists $u'\in A_{>i}$ such that $uu'\in E(G)$ and $N_G(u')\cap R(v)\cap A_i=\emptyset$.
 Observe that since $u\in A_i\cap \phi^{-1}(v)$ and $uu'\in E(G)$, by~\ref{c:witness} $u'$ has a neighbor in $R(v)\cap A_i$ in the graph $G$.
 This contradicts the supposition that $N_G(u')\cap R(v)\cap A_i=\emptyset$.
 
 As in all the cases we have obtained a contradiction, this concludes the proof of the claim. 
\end{clproof}

We now relate the optimum solution to the instance $(G,\wei)$ to optima for instances constructed for different $(x,y)\in T$.
For $(x,y)\in T$, consider a a set of instances
$$\Lambda^{x,y}\coloneqq \{\,(G^{x,y},\wei^R)\ \colon\ R\in \Rr^{x,y}\,\},$$
and let $\Lambda\coloneqq \bigcup_{(x,y)\in T} \Lambda^{x,y}$.
Note that 
$$|\Lambda|\leq |T|\cdot n^{\Oh(\Ramsey(s,t))}\leq \left( |V(H)|\cdot n \right) \cdot n^{\Oh(\Ramsey(s,t))} \leq n^{\Oh(\Ramsey(s,t))}.$$
We then have the following.

\begin{claim}\label{cl:combined-equivalent}
We have $\OPT(G,\wei)=\max_{(G',\wei')\in \Lambda} \OPT(G',\wei')$.
Moreover, for every $(G',\wei')\in \Lambda$ for which the maximum is reached, every optimum solution $\phi$ to $(G',\wei')$ is also an optimum solution to $(G,\wei)$ with $\wei(\phi)=\wei'(\phi)$.
\end{claim}
\begin{clproof}
 By \cref{cl:not-larger}, we have that 
 \begin{equation}\label{eq:not-larger}
  \OPT(G,\wei)\geq \max_{(G',\wei')\in \Lambda} \OPT(G',\wei'). 
 \end{equation}
 On the other hand, suppose $\phi^\star$ is an optimum solution to $(G,\wei)$.
 Since $T\neq \emptyset$ by assumption, hence there exists some $(x,y)\in T$ such that $\phi^\star(x)=y$.
 By \cref{cl:branching-correct}, there exists a guess $R\in \Rr^{x,y}$ such that $\phi^\star$ is compatible with $R$; note that $(G^{x,y},\wei^R)\in \Lambda$.
 By \cref{cl:not-smaller}, $\phi^\star$ is also a solution to the instance $(G^{x,y},\wei^R)$ and $\wei^R(\phi^\star)=\wei(\phi^\star)$.
 By~\eqref{eq:not-larger} we conclude that $\phi^\star$ is an optimum solution to $(G^{x,y},\wei^R)$ and $\OPT(G,\wei)=\OPT(G^{x,y},\wei^R)$.
 In particular, $\OPT(G,\wei)=\max_{(G',\wei')\in \Lambda} \OPT(G',\wei')$.
 Finally, \cref{cl:not-larger} now implies that every optimum solution to $(G^{x,y},\wei^R)$ is also an optimum solution to $(G,\wei)$.
\end{clproof}

\cref{cl:combined-equivalent} asserts that the instance $(G,\wei)$ is suitably equivalent to the set of instances $\Lambda$.
It now remains to partition each instance from $\Lambda$ into two independent subinstances $(G_1,\wei_1)$ and $(G_2,\wei_2)$ with properties required in~\ref{b:basic} and~\ref{b:progress}, 
so that the final set $\Pi$ can be obtained by applying this operation to every instance in $\Lambda$.

Consider any instance from $\Lambda$, say instance $(G^{x,y},\wei^R)$ constructed for some $(x,y)\in T$ and $R\in \Rr^{x,y}$.
We adopt the notation from the construction of $G^{x,y}$ and $\Rr^{x,y}$, and define
$$G^{x,y}_1\coloneqq G^{x,y}[A_1]\qquad\textrm{and}\qquad G^{x,y}_2\coloneqq G^{x,y}[A_2\cup A_3\cup A_4\cup X].$$
The properties of $G^{x,y}_1$ and $G^{x,y}_2$ required in~\ref{b:basic} and~\ref{b:progress} are asserted by the following claim.

\begin{claim}\label{cl:bp}
 The graphs $G^{x,y}_1$ and $G^{x,y}_2$ are $\{P_6,L_s,S_t\}$-free. 
 Moreover, for every connected graph $F$ on at least two vertices, if $G^{x,y}_1$ contains an induced $F$, then $G$ contains an induced $\univ{F}$,  and if $G^{x,y}_2$ contains an induced $F$, then~$G$ contains an induced $\ante{F}$.
\end{claim}
\begin{clproof}
 Note that $G^{x,y}_1$ is an induced subgraph of $G$.
 Moreover, $G^{x,y}_2$ is a disjoint union of $G[A_2]$, $G[A_3]$, and $G[A_4]$, plus $x_1,x_2,x_3$ are included in $G^{x,y}_2$ as isolated vertices, so every connected component of $G^{x,y}_2$ is an induced subgraph of $G$.
 As $G$ is $\{P_6,L_s,S_t\}$-free by assumption, it follows that both $G^{x,y}_1$ and $G^{x,y}_2$ are $\{P_6,L_s,S_t\}$-free.
 The second part of the statement follows directly from \cref{cl:Ai-good} and 
 the observation that every induced $F$ in $G^{x,y}_2$ has to be contained either in $G[A_2]$, or in $G[A_3]$, or in~$G[A_4]$.
\end{clproof}

Now, construct an instance $(G^{x,y}_1,\wei^R_1)$ of \textsc{Max Partial $H'$-Coloring}, where $H'=H-y$, and an instance $(G^{x,y}_2,\wei^R_2)$ of \MWLHS as follows:
$\wei^R_1$ is defined as the restriction of $\wei^R$ to the set $V(G^{x,y}_1)\times V(H')$, and $\wei^R_2$ is defined as the restriction of $\wei^R$ to the set $V(G^{x,y}_2)\times V(H)$.
Note that by \cref{cl:A1-simpler} and the construction of $\wei^R$, we have $\wei^R(u,y)=-1$ for all $u\in V(G^{x,y}_1)$, so no optimum solution to $(G^{x,y},\wei^R)$ can assign $y$ to any $u\in V(G^{x,y}_1)$.
Since in $G^{x,y}$ there are no edges between $V(G^{x,y}_1)$ and $V(G^{x,y}_2)$, we immediately obtain the following.

\begin{claim}\label{cl:separate-equivalent}
 $\OPT(G^{x,y},\wei^R)=\OPT(G^{x,y}_1,\wei^R_1)+\OPT(G^{x,y}_2,\wei^R_2)$. Moreover, for any optimum solutions $\phi_1$ and $\phi_2$ to $(G^{x,y}_1,\wei^R_1)$ and $(G^{x,y}_2,\wei^R_2)$, respectively, the function
 $\phi\coloneqq \phi_1\cup \phi_2$ is an optimum solution to $(G^{x,y},\wei^R)$.
\end{claim}

Finally, we define the set $\Pi$ to comprise of all the pairs $((G^{x,y}_1,\wei^R_1),(G^{x,y}_2,\wei^R_2))$ constructed from all $(G^{x,y},\wei^R)\in \Lambda$ as described above.
Now, assertion~\ref{b:soundness} follows directly from \cref{cl:combined-equivalent} and \cref{cl:separate-equivalent}, while assertions~\ref{b:basic} and~\ref{b:progress} are verified by \cref{cl:bp}.

It remains to argue the algorithmic aspects. There are at most $|V(H)|\cdot n=\Oh(n)$ pairs $(x,y)\in T$ to consider, and for each of them we can enumerate the set of guesses $\Rr^{x,y}$ in time $n^{\Oh(\Ramsey(s,t))}$.
It is clear that for each guess $R\in \Rr^{x,y}$, the instances $(G^{x,y}_1,\wei^R_1)$ and $(G^{x,y}_2,\wei^R_2)$ can be constructed in polynomial time. Hence the total running time of $n^{\Oh(\Ramsey(s,t))}$ follows.
This completes the proof of \cref{lem:main-branching-full}.

\paragraph*{A simplified variant.} In the next section we will rely only on the following simplified variant of \cref{lem:main-branching-full}. We provide it for the convenience of further use.

\begin{lemma}\label{lem:main-branching}
 Let $H$ be a fixed irreflexive pattern graph.
 Suppose we are given integers $s,t$ and an instance $(G,\wei)$ of \MWLHS such that $G$ is connected and $\{P_6, L_s, S_t\}$-free. 
 Denoting $n\coloneqq |V(G)|$,
 one can in time $n^{\Oh(\Ramsey(s,t))}$ construct a subgraph $G'$ of $G$ with $V(G')=V(G)$ and a set $\Pi$ consisting of at most $n^{\Oh(\Ramsey(s,t))}$ revenue functions with domain $V(G)\times V(H)$ 
 such that the following conditions hold:
\begin{enumerate}[label=(C\arabic*),ref=(C\arabic*),nosep]
 \item\label{c:progress}  The graph $G'$ is $\{P_6, L_s, S_t\}$-free. Moreover, if $G$ is $\univ{F}$-free for some connected graph $F$ on at least two vertices, then $G'$ is $F$-free.
 \item\label{c:soundness} We have $\OPT(G,\wei) = \max_{\wei' \in \Pi} \OPT(G',\wei')$. Moreover, for any $\wei'\in \Pi$ for which the maximum is reached, 
                          every optimum solution $\phi$ to $(G',\wei')$ is also an optimum solution to $(G,\wei)$ with~$\wei(\phi)=\wei'(\phi)$.
\end{enumerate}
\end{lemma}
\begin{proof}
 The proof is a simplified version of the proof of \cref{lem:main-branching-full}, hence we only highlight the differences.
 
 First, we do not iterate through all the pair $(x,y)\in T$: we perform only one construction of a subgraph $G'$ and a set of guesses $\Rr$, which is analogous to the construction of $G^{x,y}$ and $\Rr^{x,y}$ for a single 
 pair $(x,y)$ from the proof of \cref{lem:main-branching-full}.
 For $X$ we just take any set of three vertices such that $N[X]$ is a monitor in $G$, and we enumerate $X$ as $\{x_1,x_2,x_3\}$ in any way.
 The remainder of the construction proceeds as before, resulting in a family of guesses $\Rr$ of size $n^{\Oh(\Ramsey(s,t))}$ and a subgraph $G'$ of $G$ (the graph $G^{x,y}$ from the proof of \cref{lem:main-branching-full}).
 Here, in the definition of a guess we omit the condition that $\phi(x)=y$; this does not affect the asymptotic bound on the number of guesses.
 A subset of the reasoning presented in the proofs of \cref{cl:Ai-good} and \cref{cl:bp} shows that $G'$ is $\{P_6,L_s,S_t\}$-free and, moreover, for every connected graph $F$ on at least two vertices, if $G'$ contains an induced $F$, 
 then $G$ contains an induced $\univ{F}$. Note that since we are interested only in finding an induced $\univ{F}$ instead of $\ante{F}$, we do not need edges between vertices of $X$ for this. This verifies assertion~\ref{c:progress}.
 If we now define $\Pi\coloneqq \{\wei^R\colon R\in \Rr\}$, then the same reasoning as in \cref{cl:combined-equivalent} verifies assertion~\ref{c:soundness}.
 Note here that \cref{cl:not-larger} and \cref{cl:not-smaller} are still valid verbatim after replacing $G^{x,y}$ by $G'$ and $\Rr^{x,y}$ by $\Rr$.
\end{proof}

\section{Exhaustive branching}\label{sec:corollaries}

In this section we give the first set of corollaries that can be derived from \cref{lem:main-branching-full}. 
The idea is to apply this tool exhaustively, until the considered instance becomes trivial.
The main point is that with each application the clique number of the graph drops, hence we naturally obtain an upper bound of the form of $n^{f(\omega(G))}$ for the total size of the recursion tree, hence also on the running time.
This leads to results~\ref{r:P5} and~\ref{r:S3} promised in Section~\ref{sec:intro}.
In fact, we will only rely on the simplified variant of \cref{lem:main-branching-full}, that is, \cref{lem:main-branching}.

The following statement captures the idea of exhaustive applying \cref{lem:main-branching} in a recursive scheme. For the convenience of further use, we formulate the following statement so that $s$ and $t$ are given on input.

\begin{theorem}\label{thm:main-recursive}
Let $H$ be a fixed irreflexive pattern graph.
 There exists an algorithm that given $s,t\in \N$ and an instance $(G,\wei)$ of \MWLHS where $G$ is $\{P_6, L_s, S_t\}$-free, solves this instance in time $n^{\Oh(\Ramsey(s,t)\cdot \omega(G))}$.
\end{theorem}
\begin{proof}
 If $G$ is not connected, then for every connected component $C$ of $G$ we apply the algorithm recursively to $(C,\wei|_{V(C)})$. If $\phi_C$ is the computed optimum solution to this instance,
 we may output $\phi\coloneqq \bigcup_C \phi_C$. It is clear that $\phi$ constructed in this way is an optimum solution to the instance $(G,\wei)$.
 
 Assume then that $G$ is connected. If $G$ consists of only one vertex, say $u$, then we may simply output $\phi\coloneqq \{(u,v)\}$ where $v$ maximizes $\wei(u,v)$, or $\phi\coloneqq \emptyset$ if $\wei(\cdot)$ has no positive
 value in its range.
 Hence, assume that $G$ has at least two vertices, in particular $\omega(G)\geq 2$.
 We now apply \cref{lem:main-branching} to $G$. Thus, in time $n^{\Oh(\Ramsey(s,t))}$ we obtain a subgraph $G'$ of $G$ with $V(G)=V(G')$ and a suitable set of revenue functions $\Pi$ satisfying $|\Pi|\leq n^{\Oh(\Ramsey(s,t))}$.
 Recall here that $G'$ is $\{P_6,L_s,S_t\}$-free. Moreover, if we set $F=K_{\omega(G)}$ then $G$ is $\univ{F}$-free, so \cref{lem:main-branching} implies that $G'$ is $F$-free. This means that $\omega(G')<\omega(G)$.
 
 Next, for every $\wei'\in \Pi$ we recursively solve the instance $(G',\wei')$. \cref{lem:main-branching} then implies that if among the obtained optimum solutions to instances $(G',\wei')$ we pick the one with the
 largest revenue, then this solution is also an optimum solution to $(G,\wei)$ that can be output by the algorithm.
 
 We are left with analyzing the running time. Recall that every time we recurse into subproblems constructed using \cref{lem:main-branching}, the clique number of the currently considered graph drops by at least one.
 Since recursing on a disconnected graph yields connected graphs in subproblems, we conclude that the total depth of the recursion tree is bounded by $2\cdot \omega(G)$.
 In every recursion step we branch into $n^{\Oh(\Ramsey(s,t))}$ subproblems, hence the total number of nodes in the recursion tree is bounded by 
 $\left(n^{\Oh(\Ramsey(s,t))}\right)^{2\cdot \omega(G)}=n^{\Oh(\Ramsey(s,t)\cdot \omega(G))}$. 
 The internal computation in each subproblem take time $n^{\Oh(\Ramsey(s,t))}$, hence the total running time is indeed $n^{\Oh(\Ramsey(s,t)\cdot \omega(G))}$.
\end{proof}

Note that since both $L_3$ and $S_2$ contain $P_5$ as an induced subgraph, every $P_5$-free graph is $\{P_6,L_3,S_2\}$-free. Hence, from \cref{thm:main-recursive} we may immediately conclude the following statement, where
the setting of $P_5$-free graphs is covered by the case $s=3$ and $t=2$.

\begin{corollary}\label{cor:main-omega}
 For any fixed $s,t\in \N$ and irreflexive pattern graph $H$, \MWLHS can be solved in $\{P_6, L_s, S_t\}$-free graphs in time $n^{\Oh(\omega(G))}$.
 This in particular applies to $P_5$-free graphs.
\end{corollary}

Next, we observe that the statement of \cref{thm:main-recursive} can be also used for non-constant $s$ to obtain an algorithm for the case when the graph $L_s$ is not excluded.

\begin{corollary}\label{cor:main-omegat}
 For any fixed $t\in \N$ and irreflexive pattern graph $H$, \MWLHS can be solved in $\{P_6, S_t\}$-free graphs in time $n^{\Oh(\omega(G)^t)}$. 
\end{corollary}
\begin{proof}
 Observe that since the graph $L_s$ contains a clique of size $s$, every graph $G$ is actually $L_{\omega(G)+1}$-free.
 Therefore, we may apply the algorithm of \cref{thm:main-recursive} for $s\coloneqq \omega(G)+1$.
 Note here that $\omega(G)$ can be computed in time $n^{\omega(G)+\Oh(1)}$ by verifying whether $G$ has cliques of size $1,2,3,\ldots$ up to the point when the check yields a negative answer.
 Since for $s=\omega(G)+1$ and fixed $t$ we have $$\Ramsey(s,t)=\binom{s+t-2}{t-1}\leq \Oh(\omega(G)^{t-1}),$$
 the obtained running time is $n^{\Oh(\Ramsey(s,t)\cdot \omega(G))}\leq n^{\Oh(\omega(G)^t)}$.
\end{proof}

Let us note that an algorithm with running time $n^{\Oh(\omega(G)^\alpha)}$, for some constant $\alpha$, can be used within a simple branching strategy to obtain a subexponential-time algorithm.

\begin{lemma}\label{lem:subexp}
Let $H$ be a fixed irreflexive graph and suppose \MWLHS can be solved in time $n^{\Oh(\omega(G)^\alpha)}$ on $\Ff$-free graphs,
for some family of graphs $\Ff$ and some constant $\alpha \geq 1$.
Then \MWLHS can be solved in time $n^{\Oh(n^{\alpha/(\alpha+1)})}$ on $\Ff$-free graphs.
\end{lemma}
\begin{proof}
Let $(G,\wei)$ be the input instance, where $G$ has $n$ vertices.
We define threshold 
 $\tau\coloneqq \left\lfloor n^{\frac{1}{\alpha+1}}\right\rfloor$. We assume that $\tau>|V(H)|$, for otherwise the instance has constant size and can be solved in constant time.
 
 The algorithm first checks whether $G$ contains a clique on $\tau$ vertices. This can be done in time $n^{\tau+\Oh(1)}\leq n^{\Oh(n^{1/(\alpha+1)})}$ by verifying all subsets of $\tau$ vertices in $G$.
 If there is no such clique then $\omega(G)<\tau$, so we can solve the problem using the assumed algorithm in time $n^{\Oh(\omega(G)^\alpha)}\leq n^{\Oh(\tau^\alpha)}\leq n^{\Oh(n^{\alpha/(\alpha+1)})}$. 
 Hence, suppose that we have found a clique $K$ on $\tau$ vertices.
 
 Observe that since $H$ is irreflexive, in any partial $H$-coloring $\phi$ of $G$ only at most $|V(H)|$ vertices of $K$ can be colored, that is, belong to $\dom \phi$.
 We recurse into $\binom{\tau}{\leq |V(H)|}\leq n^{|V(H)|}$ subproblems: in each subproblem we fix a different subset $A\subseteq K$ with $|A|\leq |V(H)|$ and recurse on the graph $G_A\coloneqq G-(K\setminus A)$ with 
 revenue function $\wei_A\coloneqq \wei|_{V(G_A)}$. Note here that $G_A$ is $\Ff$-free.
 From the above discussion it is clear that $\OPT(G,\wei)=\max_{A\subseteq K, |A|\leq |V(H)|} \OPT(G_A,\wei_A)$.
 Therefore, the algorithm may return the solution with the highest revenue among those obtained in recursive calls.
 
 As for the running time, observe that in every recursive call, the algorithm either solves the problem in time $n^{\Oh(n^{\alpha/(\alpha+1)})}$, or recurses into $n^{|V(H)|}=n^{\Oh(1)}$ subcalls, where in each
 subcall the vertex count is decremented by at least $\left\lfloor n^{\frac{1}{\alpha+1}}\right\rfloor-|V(H)|$. It follows that the depth of the recursion is bounded by $\Oh(n^{\alpha/(\alpha+1)})$, hence the total 
 number of nodes in the recursion tree is at most $n^{\Oh(n^{\alpha/(\alpha+1)})}$. Since the time used for each node is bounded by $n^{\Oh(n^{\alpha/(\alpha+1)})}$, the total running time of $n^{\Oh(n^{\alpha/(\alpha+1)})}$
 follows.
\end{proof}

By combining \cref{cor:main-omega} and \cref{cor:main-omegat} with \cref{lem:subexp} we conclude the following.

\begin{corollary}
 For any fixed $s,t\in \N$ and irreflexive pattern graph $H$, \MWLHS can be solved in $\{P_6, L_s, S_t\}$-free graphs in time $n^{\Oh(\sqrt{n})}$. This in particular applies to $P_5$-free graphs.
\end{corollary}

\begin{corollary}
 For any fixed $t\in \N$ and irreflexive pattern graph $H$, \MWLHS can be solved in $\{P_6, S_t\}$-free graphs in time $n^{\Oh\left(n^{t/(t+1)}\right)}$.
\end{corollary}

\section{Excluding a threshold graph}\label{sec:threshold}

We now present the next result promised in~\cref{sec:intro}, namely result~\ref{r:threshold}: the problem is polynomial-time solvable on $\{P_5,F\}$-free graphs whenever $F$ is a threshold graph.
For this, we observe that a {\em{constant}} number of applications of \cref{lem:main-branching-full} reduces the input instance to instances that can be solved trivially. Thus, the whole recursion tree has polynomial size,
resulting in a polynomial-time algorithm. Note that here we use the full, non-simplified variant of \cref{lem:main-branching-full}.

We have the following statement.

\begin{theorem}\label{thm:antenna}
 Fix $s,t\in \N$.
 Suppose $F$ is a connected graph on at least two vertices such that for every fixed irreflexive pattern graph $H$, the \MWLHS problem can be solved in polynomial time in $\{P_6,L_s,S_t,F\}$-free graphs. 
 Then for every fixed irreflexive pattern graph $H$, the \MWLHS problem can be solved in polynomial time in $\{P_6, L_s, S_t,\ante{F}\}$-free graphs.
\end{theorem}
\begin{proof}
 We proceed by induction on $|V(H)|$, hence we assume that for all proper induced subgraphs $H'$ of $H$, \textsc{Max Partial $H'$-Coloring} admits a polynomial-time algorithm on $\{P_6,L_s,S_t,\ante{F}\}$-free graphs.
 Here, the base case is given by $H$ being the empty graph; then the empty function is the only solution.
 
 Let $(G,\wei)$ be an input instance $(G,\wei)$ of \MWLHS, where $G$ is $\{P_6,L_s,S_t,\ante{F}\}$-free.
 We may assume that $G$ is connected, as otherwise we may apply the algorithm to each connected component of $G$ separately, and output the union of the obtained solutions.
 Further, if the range of $\wei$ contains only non-positive numbers, then the empty function is an optimum solution to $(G,\wei)$; hence assume otherwise.
 
 We may now apply \cref{lem:main-branching-full} to $(G,\wei)$ to construct a suitable list of instances $\Pi$. Note that since $s$ and $t$ are considered fixed, $\Pi$ has polynomial size and can be computed in polynomial time.
 Consider any pair $((G_1,\wei_1),(G_2,\wei_2))\in \Pi$.
 On one hand, $(G_1,\wei_1)$ is a $\{P_6,L_s,S_t,F\}$-free instance of {\sc{Max Partial $H'$-Coloring}} where $H'$ is some proper induced subgraph of $H$, 
 so we can apply an algorithm from the inductive assumption to solve it in polynomial time.
 On the other hand, as $G$ is $\ante{F}$-free, from \cref{lem:main-branching-full} it follows that $G_2$ is $\{P_6,L_s,S_t,F\}$-free.
 Therefore, by assumption, the instance $(G_2,\wei_2)$ can be solved in in polynomial time.
 
 Finally, by \cref{lem:main-branching-full}, to obtain an optimum solution to $(G,\wei)$ it suffices
 to take the highest-revenue solution obtained as the union of optimum solutions to instances in some pair from $\Pi$. 
 As the size of $\Pi$ is polynomial and each of the instances involved in $\Pi$ can be solved in polynomial time, we can output an optimum solution to $(G,\wei)$ in polynomial time.
\end{proof}

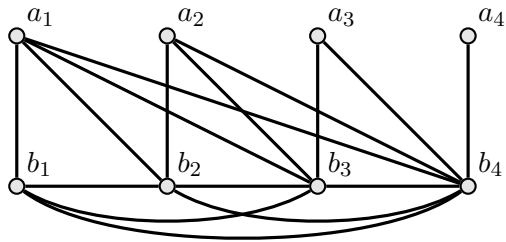
\begin{figure}[htb]
 \begin{center}
 \begin{tikzpicture}
  
   \tikzstyle{vertex}=[circle,thick,draw=black,fill=gray!20,minimum size=0.2cm,inner sep=0pt]

   \foreach \i in {1,2,3,4} {
     \node[vertex] (a\i) at (-5+2*\i,1) {};
     \node[above right] at (a\i) {$a_\i$};
     \node[vertex] (b\i) at (-5+2*\i,-1) {};
     \node[above right] at (b\i) {$b_\i$};
   }
   \foreach \i/\j in {1/1,1/2,1/3,1/4,2/2,2/3,2/4,3/3,3/4,4/4} {
     \draw[very thick] (a\i) -- (b\j);
   }
   
   \foreach \i/\j in {1/2,2/3,3/4} {
     \draw[very thick] (b\i) -- (b\j);
   }
   \foreach \i/\j in {1/3,2/4} {
     \draw[very thick] (b\i) .. controls +(1,-0.6) and +(-1,-0.6) .. (b\j);
   }
   \draw[very thick] (b1) .. controls +(1.2,-0.9) and +(-1.2,-0.9) .. (b4);
   
 \end{tikzpicture}
 \caption{The graph $Q_4$.}\label{fig:half-graph}
 \end{center}
\end{figure}   

Let us define a graph $Q_k$ as follows, see Figure~\ref{fig:half-graph}.
The vertex set consists of two disjoint sets $A\coloneqq \{a_1,\ldots,a_k\}$ and $B\coloneqq \{b_1,\ldots,b_k\}$. The set $A$ is independent in $Q_k$, while $B$ is turned into a clique.
The adjacency between $A$ and $B$ is defined as follows: for $i,j\in \{1,\ldots,k\}$, we make $a_i$ and $b_j$ adjacent if and only if $i\leq j$. Note that $Q_k$ is a threshold graph.

We now use \cref{thm:antenna} to prove the following.

\begin{corollary}\label{cor:half-graphs}
 For every fixed $k,s,t\in \N$ and irreflexive pattern graph $H$, the \MWLHS problem can be solved in polynomial time in $\{P_6,L_s,S_t,Q_k\}$-free graphs.
 This in particular applies to $\{P_5,Q_k\}$-free graphs.
\end{corollary}
\begin{proof}
 It suffices to observe that $Q_{k+1}=\ante{(Q_k)}$ and apply induction on $k$. Note that the base case for $k=1$ holds trivially, because $Q_1=K_2$, so in this case we consider the class of edgeless graphs.
 As before, the last point of the statement follows by taking $s=3$ and $t=2$ and noting that both $L_3$ and $S_2$ contain an induced $P_5$.
\end{proof}

It is straightforward to observe that for every threshold graph $F$ there exists $k\in \N$ such that $F$ is an induced subgraph of $H_k$. Therefore, from Corollary~\ref{cor:half-graphs} we can derive the following.

\begin{corollary}\label{cor:threshold}
 For every fixed threshold graph $F$ and irreflexive pattern graph $H$, \MWLHS can be solved in polynomial time in $\{P_5,F\}$-free graphs.
\end{corollary}

We now note that in Corollary~\ref{cor:half-graphs} we started the induction with $Q_1=K_2$, 
however we could also apply the reasoning starting from any other graph $F$ for which we know that \MWLHS can be solved in polynomial time in $\{P_6,L_s,S_t,F\}$-free graphs. 
One such example is $F=P_4$, for which we can derive polynomial-time solvability using a different argument.

\begin{lemma}\label{lem:P4-free}
For every fixed irreflexive pattern graph $H$, \MWLHS in $P_4$-free graphs can be solved in polynomial time.
\end{lemma}
\begin{proof}
 It is well-known that $P_4$-free graphs are exactly {\em{cographs}}, which in particular have cliquewidth at most $2$ (and a suitable clique expression can be computed in polynomial time).
 Therefore, we can solve \MWLHS in cographs in polynomial time using the meta-theorem of
 Courcelle, Makowsky, and Rotics~\cite{CourcelleMR00} for $\mathsf{MSO}_1$-expressible optimization problems on graphs of bounded cliquewidth. 
 This is because for a fixed~$H$, it is straightforward to express \MWLHS as such a problem.
 Alternatively, one can write an explicit dynamic programming algorithm, which is standard.
\end{proof}

By applying the same reasoning as in \cref{cor:half-graphs}, but starting the induction with $P_4$, we conclude:

\begin{corollary}\label{cor:gem++}
 Suppose $F$ is a graph obtained from $P_4$ by a repeated application of the $\ante{(\cdot)}$ operator.
 Then for every fixed irreflexive pattern graph $H$, \MWLHS can be solved in polynomial time in $\{P_5,F\}$-free graphs. 
\end{corollary}

\begin{figure}[htb]
 \begin{center}
 \begin{tikzpicture}
  
   \tikzstyle{vertex}=[circle,thick,draw=black,fill=gray!20,minimum size=0.2cm,inner sep=0pt]

   \begin{scope}[shift={(-3,0)}]
   \node[vertex] (x) at (0,-0.5) {};
   
   \foreach \i/\a in {1/90+36,2/90+72,3/90+108,4/90+144} {
     \node[vertex, position=\a:1.4 from x] (u\i) {};
     \draw[very thick] (x) -- (u\i);
   }
   \foreach \i/\j in {1/2,2/3,3/4} {
     \draw[very thick] (u\j) -- (u\i);
   }
   \end{scope}

   \begin{scope}[shift={(3,0)}]
   \node[vertex] (x) at (0,-0.5) {};
   
   \foreach \i/\a in {1/90+36,2/90+72,3/90+108,4/90+144} {
     \node[vertex, position=\a:1.4 from x] (u\i) {};
     \draw[very thick] (x) -- (u\i);
   }
   \foreach \i/\j in {1/2,2/3,3/4} {
     \draw[very thick] (u\j) -- (u\i);
   }
   \node[vertex, position=0:1 from x] (y) {};
   \draw[very thick] (x) -- (y);
   \end{scope}

 \end{tikzpicture}
 \caption{The gem and the graph $\ante{(P_4)}$.}\label{fig:gemante}
 \end{center}
\end{figure}
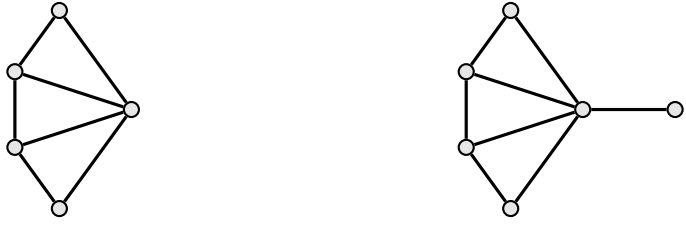   

We note here that $\ante{(P_4)}$ is the graph obtained from the {\em{gem}} graph by adding a degree-one vertex to the center of the gem; see \cref{fig:gemante}.
It turns out that $\{P_5,\textrm{gem}\}$-free graphs have bounded cliquewidth~\cite{BrandstadtLM05}, 
hence the polynomial-time solvability of \MWLHS on these graphs follows from the same argument as that used for $P_4$-free graphs in \cref{lem:P4-free}.
However, this argument does not apply to any of the cases captured by \cref{cor:gem++}.
Indeed, as shown in~\cite[Theorem 25(v)]{BrettellHMPP20}, $\{F_1,F_2\}$-free graphs have unbounded cliquewidth (and even mim-width)
whenever both $F_1$ and $F_2$ contain an independent set of size $3$, and both $P_5$ and $\ante{(P_4)}$ enjoy this property.
Note that this argument can be also applied to infer that $\{P_5,\textrm{bull}\}$-free graphs have unbounded cliquewidth and mim-width, which is the setting that we explore in the next section.

\section{Excluding a bull}\label{sec:bulls}

In this section we prove result~\ref{r:bull} promised in \cref{sec:intro}. 
The technique is similar in spirit to that used in~\cref{sec:threshold}.
Namely, we apply \cref{lem:main-branching-full} twice to reduce the problem to the case of $P_4$-free graphs, which can be handled using \cref{lem:P4-free}.
However, these applications are interleaved with a reduction to the case when the input graph is {\em{prime}}: it does not contain any non-trivial {\em{module}} (equivalently, {\em{homogeneous set}}).
This allows us to use some combinatorial results about the structure of prime $\textrm{bull}$-free graphs~\cite{ChudnovskyS18,ChudnovskyS08d}.

\subsection{Reduction to prime graphs}

In order to present the reduction to the case of prime graphs it will be convenient to work with a {\em{multicoloring}} generalization of the problem.
In this setting, we allow mapping vertices of the input graph $G$ to nonempty subset of vertices of $H$, rather than to single vertices of $H$.

\paragraph*{Multicoloring variant.} 
For a graph $H$, we write $\power(H)$ for the set of all nonempty subsets of $V(H)$.
Let $H$ be an irreflexive pattern graph and $G$ be a graph. 
A {\em{partial $H$-multicoloring}} is a partial function $\phi\colon V(G)\partto \power(H)$ that satisfies the following condition: for every edge $uu'\in E(G)$ such that $u,u'\in \dom \phi$, 
the sets $\phi(u),\phi(u')\subseteq V(H)$ are disjoint and complete to each other in $H$; that is, $vv'\in E(G)$ for all $v\in \phi(u)$ and $v'\in \phi(u')$.
We correspondingly redefine the measurement of revenue. 
A {\em{revenue function}} is a function $\wei\colon V(G)\times \power(H)\to \R$ and the revenue of a partial $H$-multicoloring $\phi$ is defined as
$$\wei(\phi)\coloneqq \sum_{u\in \dom \phi} \wei(u,\phi(u)).$$
The \MMWLHS problem is then defined as follows.

\defproblem{\MMWLHS}{Graph $G$ and a revenue function $\wei\colon V(G)\times \power(H)\to \R$}
{A partial $H$-multicoloring $\phi$ of $G$ that maximizes $\wei(\phi)$}

Clearly, \MMWLHS generalizes \MWLHS, since given an instance $(G,\wei)$ of \MWLHS, we can turn it into an equivalent instance $(G,\wei')$ of \MMWLHS by defining $\wei'$ as follows: 
for $u\in V(G)$ and $Z\subseteq V(H)$, we set
$$\wei'(u,Z)\coloneqq \begin{cases} \wei(u,v) & \textrm{if }Z=\{v\}\textrm{ for some }v\in V(H); \\ -1 & \textrm{otherwise.}\end{cases}$$
However, there is actually also a reduction in the other direction.
For an irreflexive pattern graph $H$, we define another pattern graph $\wh{H}$ as follows: $V(\wh{H})=\power(H)$ and we make $X,Y\in \power(H)$ adjacent in $\wh{H}$ if and only if $X$ and $Y$ are disjoint and complete
to each other in $H$. Note that $\wh{H}$ is again irreflexive and since we consider $H$ fixed, $\wh{H}$ is a constant-sized graph.
Then it is easy to see that the set of instances of \MMWLHS is exactly equal to the set of instances of {\textsc{Max Partial $\wh{H}$-Coloring}}, and the definitions of solutions and their revenues coincide.
Thus, we may solve instances of \MMWLHS by applying algorithms for \textsc{Max Partial $\wh{H}$-Coloring} to them.
Let us remark that expressing \MMWLHS as {\textsc{Max Partial $\wh{H}$-Coloring}} is similar to expressing $k$-tuple coloring (or fractional coloring) as homomorphisms to Kneser graphs, see e.g.~\cite[Section 6.2]{DBLP:books/daglib/0013017}.

\paragraph*{Modular decompositions.}
We are mostly interested in \MMWLHS because in this general setting, it is easy to reduce the problem once we find a non-trivial {\em{module}} (or {\em{homogeneous set}}) in an instance.
For clarity, we choose to present this approach by performing dynamic programming on a modular decomposition of the input graph, hence we need a few definitions.
The following standard facts about modular decompositions can be found for instance in the survey of Habib and Paul~\cite{HabibP10}.

A {\em{module}} (or a {\em{homogeneous set}}) in a graph $G$ is a subset of vertices $B$ such that every vertex $u\notin B$ is either complete of anti-complete to $B$.
A module $B$ is {\em{proper}} if $2\leq |B|<|V(G)|$.
A graph $G$ is called {\em{prime}} if it does not have any proper modules.

A module $B$ in a graph $G$ is {\em{strong}} if for any other module $B'$, we have either $B\subseteq B'$, or $B\supseteq B'$, or $B\cap B'=\emptyset$.
It is known that if among proper strong modules in a graph $G$ we choose the (inclusion-wise) maximal ones, then they form a partition of the vertex set of $G$, called the {\em{modular partition}} $\Mod(G)$.
The {\em{quotient graph}} $\Quo(G)$ is the graph with $\Mod(G)$ as the vertex set where two maximal proper strong modules $B,B'\in \Mod(G)$ are adjacent if they are complete to each other in $G$, and non-adjacent if
they are anti-complete to each other in $G$. It is known that for every graph $G$, the quotient graph $\Quo(G)$ is either edgeless, or complete, or prime.
Note that the quotient graph $\Quo(G)$ is always an induced subgraph of $G$: selecting one vertex from each element of $\Mod(G)$ yields a subset of vertices that induces $\Quo(G)$ in $G$.

The {\em{modular decomposition}} of a graph is a tree $\Tt$ whose nodes are modules of $G$, which is constructed by applying modular partitions recursively.
First, created a root node $V(G)$. Then, as long as the current tree has a leaf $B$ with $|B|\geq 2$, attach the elements of $\Mod(G[B])$ as children of $B$.
Thus, the leaves of $\Tt$ exactly contain all single-vertex modules of $G$; hence $\Tt$ has $n$ leaves and at most $2n-1$ nodes in total. 
It is known that the set of nodes of the modular decomposition of $G$ exactly comprises of all the strong modules in $G$.
Moreover, given $G$, the modular decomposition of $G$ can be computed in linear time~\cite{CournierH94,McConnellS94}.

\paragraph*{Dynamic programming on modular decomposition.}
The following lemma shows that given a graph~$G$, \MMWLHS in~$G$ can be solved by solving the problem for each element of $\Mod(G)$, and combining the results by solving the problem on $\Quo(G)$.
Here, $H$ is an irreflexive pattern graph that we fix from this point on.

\begin{lemma}\label{lem:modules-combine}
Let $(G,\wei)$ be an instance of \MMWLHS, where $H$ is irreflexive. For $B\in \Mod(G)$ and $W\in \power(H)$, define $\wei_{B,W}\colon B\times \power(H)\to \R$ as follows: for $u\in B$ and $Z\in \power(H)$, set
 $$\wei_{B,W}(u,Z) \coloneqq  \begin{cases}
                    \wei(u,Z) & \qquad\textrm{if }Z\subseteq W;\\
                    -1        & \qquad \textrm{otherwise.}
                   \end{cases}$$
Further, define $\wei'\colon \Mod(G)\times \power(H)\to \R$ as follows: for $B\in \Mod(G)$ and $W\in \power(H)$, set
 $$\wei'(B,W)\coloneqq \OPT(G[B],\wei_{B,W}).$$
Then $\OPT(G,\wei)=\OPT(\Quo(G),\wei')$. Moreover, for every optimum solution $\phi'$ to $(\Quo(G),\wei')$ and optimum solutions $\phi_B$ to respective instances $(G[B],\wei_{B,\phi'(B)})$, for $B\in \Mod(G)\cap \dom \phi'$,
the function
$$\phi\coloneqq \bigcup_{B\in \Mod(G)\cap \dom \phi'} \phi_B$$
is an optimum solution to $(G,\wei)$.
\end{lemma}
\begin{proof}
 We first argue that $\OPT(G,\wei)\leq \OPT(\Quo(G),\wei')$. 
 Take an optimum solution $\phi$ to $(G,\wei)$. For every $B\in \Mod(G)$, let 
 $$
   \phi'(B)\coloneqq \bigcup_{u\in B\cap \dom \phi} \phi(u),
 $$
 unless the right hand side is equal to $\emptyset$, in which case we do not include $B$ in the domain of $\phi'$.
 Observe that $\phi'$ defined in this manner is a solution to the instance $(\Quo(G),\wei')$. 
 Indeed, if for some $BB'\in E(\Quo(G))$ we did not have that $\phi'(B)$ and $\phi'(B')$ are disjoint and complete to each other in $H$, then there would exist $u\in B$ and $u'\in B'$ such that
 $\phi(u)$ and $\phi(u')$ are not disjoint and complete to each other in $H$, contradicting the assumption that $\phi$ is a solution to $(G,\wei)$. 
 
 Note that for each $B\in \dom \phi'$, $\phi|_B$ is a solution to the instance $(G[B],\wei_{B,\phi'(B)})$. Observe that
 $$
   \OPT(G,\wei) = \wei(\phi) = \sum_{B\in \dom \phi'} \wei_{B,\phi'(B)}(\phi|_B)\leq \sum_{B\in \dom \phi'} \OPT(G[B],\wei_{B,\phi'(B)}),
 $$
 where the second equality follows from the fact that $\wei$ and $\wei_{B,\phi'(B)}$ agree on all pairs $(u,\phi(u))$ for $u\in B\cap \dom \phi$.
 On the other hand, since $\phi'$ is a solution to $(\Quo(G),\wei')$, we have
 $$
   \sum_{B\in \dom \phi'} \OPT(G[B],\wei_{B,\phi'(B)}) = \sum_{B\in \dom \phi'} \wei'(B,\phi'(B)) = \wei'(\phi') \leq \OPT(\Quo(G),\wei').
 $$
 This proves that $\OPT(G,\wei)\leq \OPT(\Quo(G),\wei')$.
 
 Next, we argue that $\OPT(G,\wei)\geq \OPT(\Quo(G),\wei')$ and that the last assertion from the lemma statement holds.
 Let $\phi'$ be an optimum solution to the instance $(\Quo(G),\wei')$. 
 Further, for each $B\in \dom \phi'$, let $\phi_B$ be any optimum solution to the instance $(G[B],\wei_{B,\phi'(B)})$.
 Consider 
 $$
   \phi\coloneqq \bigcup_{B\in \dom \phi'} \phi_B
 $$
 We verify that $\phi$ is a solution to $(G,\wei)$. 
 The only non-trivial check is that for any $B,B'\in \dom \phi'$ with $BB'\in E(\Quo(G))$, $u\in \dom \phi_B$, and $u'\in \dom \phi_{B'}$, we have that $\phi(u)$ and $\phi(u')$ are disjoint and complete to each other in $H$.
 However, $\phi_B$, as an optimal solution to $(G[B],\wei_{B,\phi'(B)})$, does not use any assignments with negative revenues, which implies that $\phi(u)=\phi_B(u)\subseteq \phi'(B)$.
 Similarly, we have $\phi(u')=\phi_{B'}(u')\subseteq \phi'(B')$. Since $\phi'(B)$ and $\phi'(B')$ are disjoint and complete to each other, due to the assumption that $\phi'$ is a solution to $(\Quo(G),\wei')$, the
 same can be also claimed about $\phi(u)$ and $\phi(u')$.
 
 Finally, observe that 
 $$
   \wei(\phi)  =  \sum_{B\in \dom \phi'} \wei_{B,\phi'(B)}(\phi_B) =  \sum_{B\in \dom \phi'} \OPT(G[B],\wei_{B,\phi'(B)})=\wei'(\phi')=\OPT(\Quo(G),\phi'),
 $$
 where the first equality follows from the fact that $\wei$ and $\wei_{B,\phi'(B)}$ agree on all assignments used by~$\phi$, for all $B\in \dom \phi'$.
 This proves that $\OPT(G,\wei)\geq \OPT(\Quo(G),\wei')$. By combining this with the reverse inequality proved before, we conclude that
 $\OPT(G,\wei)=\OPT(\Quo(G),\wei')$ and $\phi$ is an optimum solution to $(G,\wei)$.
\end{proof}

\cref{lem:modules-combine} enables us to perform dynamic programming on a modular decomposition, provided the problem can be solved efficiently on prime graphs from the considered graph class.
This leads to the following statement.

\begin{lemma}\label{lem:reduction-prime}
Let $H$ be a fixed irreflexive pattern graph.
 Let $\Ff$ be a set of graphs such that \MMWLHS can be solved in time $T(n)$ on prime $\Ff$-free graphs. Then \MMWLHS can be solved in time $n^{\Oh(1)}\cdot T(n)$ on $\Ff$-free graphs.
\end{lemma}
\begin{proof}
 First, in linear time we compute the modular decomposition $\Tt$ of $G$.
 Then, for every strong module $B$ of $G$ and every $W\in \power(H)$, we will compute an optimum solution $\phi_{B,W}$ to the instance $(G[B],\wei_{B,W})$, 
 where the revenue function $\wei_{B,W}$ is defined as in \cref{lem:modules-combine}. At the end, we may return $\phi_{V(G),V(H)}$ as the optimum solution to $(G,\wei)$.
 
 The computation of solutions $\phi_{B,W}$ is organized in a bottom-up manner over the decomposition $\Tt$. Thus, whenever we compute solution $\phi_{B,W}$ for a strong module $B$ and $W\in \power(H)$, 
 we may assume that the solutions $\phi_{B',W'}$ for all $B'\in \Mod(G[B])$ and $W'\in \power(H)$ have already been computed.

 When $B$ is a leaf of $\Tt$, say $B=\{u\}$ for some $u\in V(G)$, then for every $W\in \power(W(H))$ we may simply output $\phi_{B,W}\coloneqq \{(u,Z)\}$ where $Z$ maximizes $\wei_{B,W}(u,Z)$, 
 or $\phi_{B,W}\coloneqq \emptyset$ if $\wei_{B,W}$ has no positive values in its range.
 
 Now suppose $B$ is a non-leaf node of $\Tt$ and $W\in \power(W(H))$. 
 Construct an instance $(\Quo(G[B]),\wei')$ similarly as in the statement of \cref{lem:modules-combine}: for $B'\in \Mod(G[B])$ and $Z\in \power(H)$, we put
 $$\wei'(B,W)\coloneqq \OPT(G[B'],\wei_{B',W\cap Z}).$$
 Note here that the values $\OPT(G[B'],\wei_{B',W\cap Z})$ have already been computed, as they are equal to $\wei_{B',W\cap Z}(\phi_{B',W\cap Z})$.
 From \cref{lem:modules-combine} applied to the instance $(G[B],\wei_{W,B})$ it follows that if $\phi'$ is an optimum solution to the instance $(\Quo(G[B]),\wei')$, then the union of solutions
 $\phi_{B',\phi'(B')}$ over all $B'\in \dom \phi'$ is an optimum solution to $(G[B],\wei_{B,W})$. Therefore, it remains to solve the instance $(\Quo(G[B]),\wei')$.
 We make a case distinction depending on whether $\Quo(G[B])$ is edgeless, complete, or prime.
 
 It is very easy to argue that \MMWLHS can be solved in polynomial time both in edgeless graphs and in complete graphs.
 For instance, one can equivalently see the instance as an instance of \textsc{Max Partial $\wh{H}$-Coloring}, and apply the algorithm for $P_4$-free graphs given by \cref{lem:P4-free}.
 
 On the other hand, if $\Quo(G[B])$ is prime, then by assumption we can solve the instance $(\Quo(G[B]),\wei')$ in time $T(n)$. Recall here that $\Quo(G[B])$ is an induced subgraph of $G[B]$, hence it is also $\Ff$-free.
 
 This concludes the description of the algorithm.
 As for the running time, observe that since $H$ is considered fixed, the computation for each node of the decomposition take time $n^{\Oh(1)}\cdot T(n)$.
 Since $\Tt$ has at most $2n-1$ nodes, the total running time of $n^{\Oh(1)}\cdot T(n)$ follows.
\end{proof}

We can now conclude the following statement. Note that it speaks only about the standard variant of the \MWLHS problem.

\begin{theorem}\label{thm:reduce-prime}
 Let $\Ff$ be a set of graphs such that for every fixed irreflexive pattern graph $H$, the \MWLHS problem can be solved in polynomial time in prime $\Ff$-free graphs.
 Then for every fixed irreflexive pattern graph $H$, the \MWLHS problem can be solved in polynomial time in $\Ff$-free graphs.
\end{theorem}
\begin{proof}
 As instances of \MMWLHS can be equivalently regarded as instances of \textsc{Max Partial $\wh{H}$-Coloring}, we conclude that for every fixed $H$, \MMWLHS is polynomial-time solvable in prime $\Ff$-free graphs ---
 just apply the algorithm for \textsc{Max Partial $\wh{H}$-Coloring}. By \cref{lem:reduction-prime} we infer that for every fixed $H$, \MMWLHS is polynomial-time solvable in $\Ff$-free graphs.
 As \MMWLHS generalizes \MWLHS, this algorithm can be used to solve \MWLHS in $\Ff$-free graphs in polynomial time.
\end{proof}

\subsection{Algorithms for bull-free classes}

We now move to our algorithmic results for subclasses of $\textrm{bull}$-free graphs. 
For this, we need to recall some definitions and results.

For graphs $F$ and $G$, we say that $G$ contains an {\em{induced $F$ with a center and an anti-center}} if there exists $A\subseteq V(G)$ such that $G[A]$ is isomorphic to $F$, and moreover there are vertices
$x,y\notin A$ such that $x$ is complete to $A$ and $y$ is anti-complete to $A$.
Observe that if a graph $G$ contains an induced $\ante{F}$, then $G$ contains an induced $F$ with a center and an anti-center.
We will use the following.

\begin{theorem}[\cite{ChudnovskyS18}]\label{thm:bull-P4}
 Let $G$ be a $\{\textrm{bull}, C_5\}$-free graph. If $G$ contains an induced $P_4$ with a center and
an anti-center, then $G$ is not prime.
\end{theorem}

\begin{theorem}[\cite{ChudnovskyS08d}]\label{thm:bull-C5}
 Let $G$ be a $\textrm{bull}$-free graph. If $G$ contains an induced $C_5$ with a center and
an anti-center, then $G$ is not prime.
\end{theorem}

We now combine \cref{lem:main-branching}, \cref{thm:reduce-prime}, and \cref{thm:bull-P4} to show the following.

\begin{lemma}\label{lem:bull-no-C5}
 For every fixed $t\in \N$ and irreflexive pattern graph $H$, the \MWLHS problem in $\{P_6,C_5,S_t,\textrm{bull}\}$-free graphs can be solved in polynomial time. 
\end{lemma}
\begin{proof}
 As in the proof of \cref{thm:antenna}, we proceed by induction on $|V(H)|$.
 Hence, we assume that for all proper induced subgraphs $H'$ of $H$, \textsc{Max Partial $H'$-Coloring} can be solved in polynomial-time on $\{P_6,C_5,S_t,\textrm{bull}\}$-free graphs.
 By \cref{thm:reduce-prime}, it suffices to give a polynomial-time algorithm for \MWLHS working on prime $\{P_6,C_5,S_t,\textrm{bull}\}$-free graphs.
 By \cref{thm:bull-P4}, such graphs do not contain any induced $P_4$ with a center and an anti-center, so in particular they do not contain any induced $\ante{(P_4)}$.
 
 Consider then an input instance $(G,\wei)$ of \MWLHS, where $G$ is $\{P_6,C_5,S_t,\textrm{bull}\}$-free and prime, hence also connected.
 If the range of $\wei$ consists only of non-positive numbers, then the empty function is an optimum solution to $(G,\wei)$, hence assume otherwise.
 Note that $L_3$ contains an induced bull, hence we may apply \cref{lem:main-branching} for $s=3$ to compute a suitable set $\Pi$ of pairs of instances. This takes polynomial time due to $t$ being considered a constant.
 
 Consider any pair $((G_1,\wei_1),(G_2,\wei_2))\in \Pi$.
 On one hand, $(G_1,\wei_1)$ is an instance of {\sc{Max Partial $H'$-Coloring}} for some proper induced subgraph $H'$ of $H$, hence we can apply an algorithm from the inductive assumption to solve it in polynomial time.
 On the other hand, note that the graph $G_2$ is $P_4$-free, for if it had an induced $P_4$, 
 then by \cref{lem:main-branching} we would find an induced $\ante{(P_4)}$ in $G$, a contradiction to $G$ being prime by \cref{thm:bull-P4}.
 Hence, we can solve the instance $(G_2,\wei_2)$ in polynomial time using the algorithm of \cref{lem:P4-free}.
 
 Finally, \cref{lem:main-branching} implies that to obtain an optimum solution to $(G,\wei)$, it suffices
 to take the highest-revenue solution obtained as the union of optimum solutions to instances in some pair from $\Pi$. 
 Since the size of $\Pi$ is polynomial and each of the instances involved in $\Pi$ can be solved in polynomial time, we can output an optimum solution to $(G,\wei)$ in polynomial time as well.
\end{proof}

Finally, it remains to combine \cref{lem:bull-no-C5} with \cref{lem:main-branching} again to derive the main result of this section.

\begin{theorem}\label{thm:main-bull}
 For every fixed $t\in \N$ and irreflexive pattern graph $H$, the \MWLHS problem in $\{P_6,S_t,\textrm{bull}\}$-free graphs can be solved in polynomial time. 
\end{theorem}
\begin{proof}
 We follow exactly the same strategy as in the proof of \cref{lem:bull-no-C5}. The differences are that:
 \begin{itemize}[nosep]
  \item Instead of using \cref{thm:bull-P4}, we apply \cref{thm:bull-C5} to argue that the graph $G_2$ is $C_5$-free.
  \item Instead of using \cref{lem:P4-free} to solve $P_4$-free instances, we apply \cref{lem:bull-no-C5} to solve $\{P_6,C_5,S_t,\textrm{bull}\}$-free instances.
 \end{itemize}
 The straightforward application of these modifications is left to the reader.
\end{proof}

Finally, since $S_2=P_5$, from \cref{thm:main-bull} we immediately conclude the following.

\begin{corollary}
 For every fixed irreflexive pattern graph $H$, the \MWLHS problem in $\{P_5,\textrm{bull}\}$-free graphs can be solved in polynomial time.
\end{corollary}

\section{Hardness for patterns with loops}\label{sec:hardness}

Recall that the assumption that $H$ is irreflexive is crucial in our approach in \cref{lem:main-branching-full}.
However, while \textsc{$H$-Coloring} becomes trivial if $H$ has loops, this is no longer the case for generalizations of the problem,
including  \LHomo and \MWLHS. See e.g. \cite{DBLP:journals/jgt/FederHH03,DBLP:journals/ejc/GutinHRY08,DBLP:journals/jcss/OkrasaR20}.

Here, \LHomo is the list variant of the {\sc{$H$-Coloring}} problem: an instance of \LHomo is a pair $(G,L)$, where $G$ is a graph and $L \colon V(G) \to 2^{V(H)}$ assigns a \emph{list} to every vertex.
We ask whether $G$ admits an $H$-coloring $\phi$ that respects lists $L$, i.e., $\phi(v) \in L(v)$ for every $v \in V(G)$.

Note that that \LHomo is a special case of \MWLHS: for any instance $(G,L)$ of \LHomo, define the revenue function $\wei \colon V(G) \times V(H) \to \mathbb{R}$ as follows:
\[\wei(v,u)=\begin{cases} -1 &\quad  \textrm{if } u \notin L(v);\\ 1 &\quad  \textrm{if } u \in L(v).\end{cases}\]
It is straightforward to observe that solving the instance $(G,L)$ of \LHomo is equivalent to deciding if the instance $(G,\wei)$ of \MWLHS has a solution of revenue at least (in fact, equal to) $|V(G)|$.
Thus any positive result for \MWLHS can be applied to \LHomo, while any hardness result for \LHomo carries over to \MWLHS.

Let us point out that if we only aim for solving \textsc{List $H$-Coloring}, a simple adaptation of the algorithm of  Ho\`ang et al.~\cite{kcolp5free} 
shows that the problem is polynomial-time solvable in $P_5$-free graphs, provided $H$ has no loops.
In this section we show that there is little hope to extend this positive result to graphs $H$ with loops allowed.

A graph $G$ is a \emph{split graph} if $V(G)$ can be partitioned into a clique and an independent set (that we call the \emph{independent part}).
Is is well-known that split graphs are precisely $\{P_5,C_4,2P_2\}$-free graphs.

Let $H_0$ be the graph on the vertex set $\bigcup_{i \in \{1,2,3\}} \{a_i,b_i,c_i,d_i\}$  (see \cref{fig:hard}).
The edge set $E(H_0)$ consists of the edges:
\begin{itemize}[nosep]
\item all edges with both endpoints in $\bigcup_{i \in \{1,2,3\}} \{a_i,b_i\}$ (including loops),
\item all edges with both endpoints in $\bigcup_{i \in \{1,2,3\}} \{c_i,d_i\}$ (including loops),
\item for each $i \in \{1,2,3\}$, the edges $a_ic_i$ and $b_ic_i$,
\item for each $i \in \{1,2,3\}$ and $j \in \{1,2,3\} \setminus \{i\}$, the edges $d_ia_i$ and $d_ib_j$.
\end{itemize}

\begin{figure}[htb]
 \begin{center}
\includegraphics[scale=1.0]{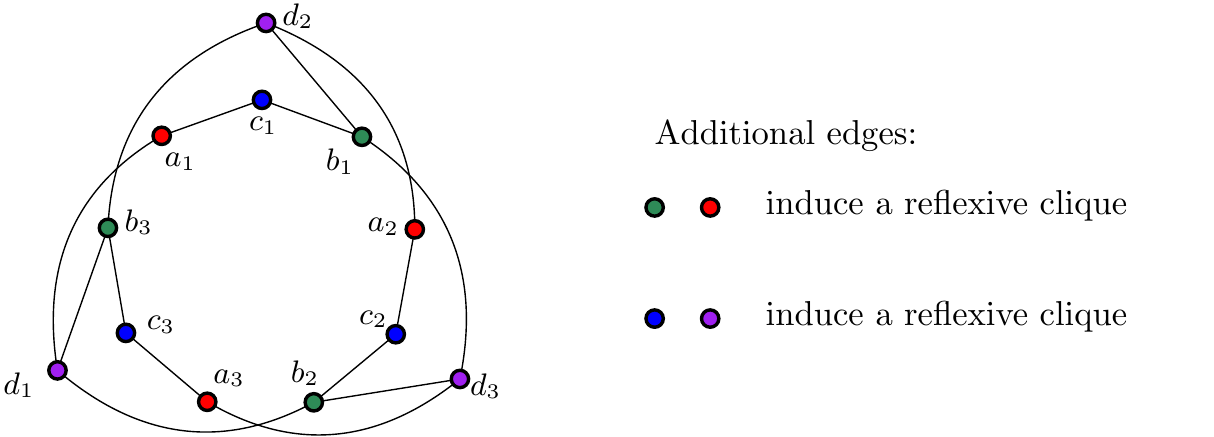}
 \caption{The graph $H_0$ used in \cref{thm:hard}.}\label{fig:hard}
 \end{center}
\end{figure}

\begin{theorem}\label{thm:hard}
The \textsc{List $H_0$-Coloring} problem (and thus \textsc{Max Partial $H_0$-Coloring}) is $\mathsf{NP}$-hard and, under the ETH, cannot be solved in time $2^{o(n)}$:
\begin{enumerate}[label=(\alph*),ref=(\alph*),nosep]
\item in split graphs, even if each vertex of the independent part is of degree 2; and\label{stat:hardonsplit}
\item in complements of bipartite graphs (in particular, in $\{P_5,\textrm{bull}\}$-free graphs).\label{stat:hardoncobipartite}
\end{enumerate}
\end{theorem}
\begin{proof}
We partition the vertices of $H_0$ into sets $A,B,C,D$, where $A \coloneqq \{a_1,a_2,a_3\}$ and the remaining sets are defined analogously.

We reduce from \textsc{3-Coloring}, which is $\mathsf{NP}$-complete and cannot be solved in time $2^{o(n+m)}$ unless the ETH fails,
where $n$ and $m$ respectively denote the number of vertices and of edges~\cite{DBLP:books/sp/CyganFKLMPPS15}.
Let $G$ be an instance of \textsc{3-Coloring} with $n$ vertices and $m$ edges. Let $V(G) = \{v_1,v_2,\ldots,v_n\}$ and let $[n]\coloneqq \{1,\ldots,n\}$.

First, let us build a split graph $G'$ with lists $L$, which admits an $H_0$-coloring respecting $L$ if and only if $G$ is 3-colorable.
For each $i \in [n]$, we add to $G'$ two vertices $x_i$ and $y_i$.
Let $X \coloneqq \{x_i \colon i \in [n]\}$ and $Y \coloneqq \{y_i \colon i \in [n]\}$.
We make $X \cup Y$ into a clique in $G'$. We set $L(x_i) \coloneqq \{a_1,a_2,a_3\}$ and $L(y_i)\coloneqq \{b_1,b_2,b_3\}$ for every $i \in [n]$.

The intended meaning of an $H_0$-coloring of $G'$ is that for any $i \in [n]$ and $j \in \{1,2,3\}$,
coloring $x_i$ with color $a_j$ and $y_i$ with color $b_j$ corresponds to coloring $v_i$ with color $j$.
So we need to ensure the following two properties:
\begin{enumerate}[label=(P\arabic*),ref=P\arabic*,nosep]
\item for every $i \in [n]$ and $j \in \{1,2,3\}$, the vertex $x_i$ is colored $a_j$ if and only if the vertex $y_i$ is colored $b_j$,\label{phard:1}
\item for every edge $v_iv_j$ of $G$, the vertices $x_i$ and $x_j$ get different colors (and, by \cref{phard:1}, so do $y_i$ and $y_j$).\label{phard:2}
\end{enumerate}
In order to ensure property \cref{phard:1}, for each $i \in [n]$ we introduce a vertex $w_i$, adjacent to $x_i$ and $y_i$, whose list is $\{c_1,c_2,c_3\}$. By $W$ we denote the set $\{w_i \colon i \in [n]\}$.
To ensure property \cref{phard:2}, for each edge $v_iv_j$ of $G$, where $i < j$, we introduce a vertex $z_{i,j}$ adjacent to $x_i$ and $y_j$. 
The list of $z_{i,j}$ consists of $\{d_1,d_2,d_3\}$. By $Z$ we denote the set $\{z_{i,j} \colon v_iv_j \in E(G) \text{ and } i<j\}$.

It is straightforward to verify that the definition of the neighborhoods of vertices $c_i, d_i$ in $H_0$ forces \cref{phard:1} and \cref{phard:2}, 
which implies that $G$ is 3-colorable if and only if $G'$ admits an $H_0$-coloring that respects lists~$L$.
The number of vertices of $G'$ is
\[
 |X|+|Y|+|W|+|Z|=n+n+n+m=\Oh(n+m).
\]
Hence, if the obtained instance of \textsc{List $H_0$-Coloring} could be solved in time $2^{o(|V(G')|)}$, 
then this would imply the existence of a $2^{o(n+m)}$-time algorithm for {\sc{$3$-Coloring}}, a contradiction with the ETH.
Furthermore, $X \cup Y$ is a clique, $W \cup Z$ is independent, and every vertex from $W \cup Z$ has degree 2. Thus the statement~\ref{stat:hardonsplit} of the theorem holds.

We observe that the set $\{L(v) \colon v \in W \cup Z\} = C \cup D$  forms a reflexive clique in $H_0$. Thus we can turn the set $W \cup Z$ into a clique, obtaining an equivalent instance $(G'',L)$ of \textsc{List $H_0$-Coloring}. As the vertex set of $G''$ can be partitioned into two cliques, $G''$ is the complement of a bipartite graph, so the statement~\ref{stat:hardoncobipartite} of the theorem holds as well.
\end{proof}

\section{Open problems}\label{sec:conclusions}

The following question, which originally motivated our work, still remains unresolved.

\begin{question}
 Is there a polynomial-time algorithm for \OCT in $P_5$-free graphs?
\end{question}

Note that our work stops short of giving a positive answer to this question: we give an algorithm with running time $n^{\Oh(\omega(G))}$, a subexponential-time algorithm, and polynomial time algorithms for
the cases when either a threshold graphs or a bull is additionally forbidden. Therefore, we are hopeful that the answer to the question is indeed positive.

\medskip

One aspect of our work that we find particularly interesting is the possibility of treating the clique number $\omega(G)$ as a progress measure for an algorithm, which enables
bounding the recursion depth in terms of $\omega(G)$. This approach naturally leads to algorithms with running time of the form $n^{f(\omega(G))}$ for some function $f$, that is, polynomial-time for every fixed clique number.
By \cref{lem:subexp}, having a polynomial function $f$ in the above implies the existence of a subexponential-time algorithm, at least in the setting of \MWLHS for irreflexive $H$.
However, looking for algorithms with time complexity $n^{f(\omega(G))}$ seems to be another relaxation of the goal of polynomial-time solvability, 
somewhat orthogonal to subexponential-time algorithms~\cite{BacsoLMPTL19,Brause17,GroenlandORSSS19} or approximation schemes~\cite{ChudnovskyPPT20}.
Note that our work and the recent work of Brettell et al.~\cite{BrettellHP20} actually show two different methods of obtaining such algorithms: using direct recursion, or via dynamic programming on branch decompositions of 
bounded mim-width.
It would be interesting to investigate this direction in the context of {\sc{Maximum Independent Set}} in $P_t$-free graphs.
A concrete question would be the following.

\begin{question}
 Is there a polynomial-time algorithm for {\sc{Maximum Independent Set}} in $\{P_t,K_t\}$-free graphs, for every fixed $t$?
\end{question}

In all our algorithms, we state the time complexity assuming that the pattern graph $H$ is fixed. This means that the constants hidden in the $\Oh(\cdot)$ notation in the exponent may --- and do --- depend on the size of $H$.
In the language of parameterized complexity, this means that we give $\mathsf{XP}$ algorithms for the parameterization by the size of $H$. 
It is natural to ask whether this state of art can be improved to the existence of $\mathsf{FPT}$ algorithms, that is, with running time $f(H)\cdot n^c$ for some computable function $f$ and universal constant $c$, independent of $H$.
This is not known even for the case of {\sc{$k$-Coloring}} $P_5$-free graphs, so let us re-iterate the old question of Ho\`ang et al.~\cite{kcolp5free}.

\begin{question}
 Is there an $\mathsf{FPT}$ algorithm for {\sc{$k$-Coloring}} in $P_5$-free graphs parameterized by $k$?
\end{question}

While the above question seems hard, it is conceivable that $\mathsf{FPT}$ results could be derived in some more restricted settings considered in this work, for instance for 
$\{P_5,\textrm{bull}\}$-free graphs.

\medskip

Finally, recall that \LHomo in $P_5$-free graphs is polynomial-time solvable for irreflexive $H$, but might become $\mathsf{NP}$-hard when loops on $H$ are allowed (see \cref{thm:hard}). 
We believe that it would be interesting to obtain a full complexity dichotomy.

\begin{question}
For what pattern graphs $H$ (with possible loops) is \textsc{List $H$-Coloring} polynomial-time solvable in $P_5$-free graphs?
\end{question}

We think that solving all problems listed above might require obtaining new structural results, and thus may lead to better understanding of the structure of $P_5$-free graphs.

\paragraph*{Acknowledgements.} We acknowledge the welcoming and productive atmosphere at Dagstuhl Seminar~19271 ``Graph Colouring: from Structure to Algorithms'', where this work has been initiated.

\end{document}